\documentclass{article}

\usepackage{arxiv}

\usepackage[utf8]{inputenc} 
\usepackage[T1]{fontenc}    
\usepackage{hyperref}       
\usepackage{url}            
\usepackage{booktabs}       
\usepackage{amsfonts}       
\usepackage{nicefrac}       
\usepackage{microtype}      
\usepackage{lipsum}		
\usepackage{graphicx}
\usepackage{natbib}
\usepackage{doi}
\usepackage{amsmath}
\usepackage{mathtools, nccmath}

\usepackage{semantic}
\usepackage{proof}
\usepackage{stmaryrd}

\usepackage{enumitem}
\usepackage{booktabs}
\usepackage{array}
\usepackage{multirow}
\usepackage{graphicx}
\usepackage{listings}
\usepackage{footmisc}
\usepackage{xcolor}
\usepackage{subcaption}

\usepackage[export]{adjustbox}
\usepackage{mdframed}

\usepackage[ruled]{algorithm2e}

\newtheorem{theorem}{Theorem}
\newtheorem{proof}{Proof}

 \newcounter{Q}



\usepackage{tikz}
\usetikzlibrary{automata,positioning,matrix, shapes}

\lstdefinelanguage{CypherSQL}{
  morekeywords={MATCH, RETURN , DISTINCT, SELECT, 
  FROM, JOIN, AS, ON, UNION, INSERT, DROP, ALTER, 
  CREATE, END, LOOP, BEGIN, TABLE, RENAME,
  WHILE, EXISTS,TEMPORARY, RECURSIVE, VIEW},
  sensitive=false, 
  morecomment=[l]{//}, 
  morecomment=[s]{/*}{*/}, 
  morestring=[b]" 
} %
\lstdefinelanguage{EXEPlan}{
    morekeywords={HashAggregate, cost, rows, Group, Key, Planned,
    Partitions, Hash, Join, width, Cond, Seq, Scan, on, Merge,
    Index, Only, using, Sort},
    sensitive=false, 
    morecomment=[l]{//}, 
    morecomment=[s]{/*}{*/}, 
    morestring=[b] 
} %

\definecolor{codegreen}{rgb}{0,0.6,0}
\definecolor{codegray}{rgb}{0.5,0.5,0.5}
\definecolor{codepurple}{rgb}{0.58,0,0.82}
\definecolor{backcolour}{rgb}{1.0,1.0,1.0}
\definecolor{blu}{rgb}{0.0, 0.0, 1.0}
\definecolor{framegrey}{rgb}{0.52, 0.52, 0.51}

\lstdefinestyle{CypherSQL}{
    backgroundcolor=\color{backcolour},   
    commentstyle=\bfseries\color{codegreen},
    keywordstyle=\small\color{blu},
    numberstyle=\bfseries\small\color{black},
    stringstyle=\color{codepurple},
    basicstyle=\footnotesize\ttfamily,
    breakatwhitespace=true,         
    breaklines=false,                 
    captionpos=b,                    
    keepspaces=false,                 
    numbers=left,                    
    numbersep=1pt,
    xleftmargin=.5cm,                  
    showspaces=false,                
    showstringspaces=true,
    showtabs=false,                  
    tabsize=2
}

\lstdefinestyle{exeplan}{
    backgroundcolor=\color{backcolour},   
    commentstyle=\bfseries\color{codegreen},
    keywordstyle=\color{codepurple},
    numberstyle=\bfseries\small\color{black},
    stringstyle=\color{codepurple},
    basicstyle=\footnotesize\ttfamily, 
    columns=flexible,
    breakatwhitespace=true,         
    breaklines=true,                 
    captionpos=b,                    
    keepspaces=false,                 
    numbers=left,                    
    numbersep=1pt,
    xleftmargin=.5cm,                  
    showspaces=false,                
    showstringspaces=true,
    showtabs=false,                 
    tabsize=2
}

\newcommand{\Schema}[1]{\mathsf{S}_{#1}}

\newcommand{\Node}[1]{\mathsf{N}_{#1}}
\newcommand{\Edge}[1]{\mathsf{E}_{#1}}
\newcommand{\Property}[1]{\mathsf{P}_{#1}}
\newcommand{\Key}[1]{\mathsf{K}_{#1}}
\newcommand{\Type}[1]{\mathsf{T}_{#1}}
\newcommand{\Value}[1]{\mathsf{V}_{#1}}
\newcommand{\Label}[1]{\mathsf{L}_{#1}}
\newcommand{\ValueToType}[1]{\Upsilon_{#1}}
\newcommand{\Database}[1]{\mathsf{D}_{#1}}

\newcommand\edgelabelvar{l_e}
\newcommand\nodelabelvar{l_n}

\newcommand{\SchemaDatabaseMap}[1]{\mathcal{SD}_{#1}}

\newcommand*\sem[2][\Database{}]{\llbracket {#2} \rrbracket_{#1}}

\newcommand{\NodeLabeling}[1]{\eta_{#1}}
\newcommand{\EdgeLabeling}[1]{\xi_{#1}}
\newcommand{\PropertyLabeling}[1]{\Delta_{#1}}
\newcommand{\EdgeSrcTrgLabeling}[1]{\lambda_{#1}}

\newcommand\nodesch{\Node{\Schema{}}}
\newcommand\edgesch{\Edge{\Schema{}}}
\newcommand\nodelab{\NodeLabeling{\Schema{}}}
\newcommand\edgelab{\EdgeLabeling{\Schema{}}}
\newcommand\edgesrctrg{\EdgeSrcTrgLabeling{\Schema{}}}

\newcommand{\SchemaTriple}[1]{\mathcal{T}_{#1}}
\newcommand{\tripleSrc}[1]{\texttt{sc}_{#1}}
\newcommand{\tripleTrg}[1]{\texttt{tr}_{#1}}
\newcommand{\tripleEdgeLabel}[1]{\texttt{eT}_{#1}}

\newcommand{\CQTQuery}[1]{\mathsf{C}_{#1}}
\newcommand{\UCQTQuery}[1]{\mathsf{UC}_{#1}}

\newcommand{\PlusComp}[1]{\mathsf{PlC}_{#1}}

\newcommand\BasicSchemaTriples{\SchemaTriple{b}}
\newcommand\triplevar{t}
\newcommand*\SingleTriple[3]{(#1, #2, #3)}
\newcommand*\judgement[3]{\vdash_{#1} {#2}\colon {#3}}
\newcommand*\judg[2]{\judgement{\Schema{}}{#1}{#2}}
\newcommand*\typeOf[2]{\mathcal{T}_{#1}(#2)}
\newcommand*\typ[1]{\typeOf{\Schema{}}{#1}}

\newcommand*\annquery[1]{\psi_{#1}}
\newcommand*\annconcat[1]{/_{#1}}

\newcommand*\queryOf{\mathcal{Q}}

\newtheorem{definition}{Definition}
\newtheorem{example}{Example}

\newcommand{\concatenation}{/}
\newcommand*\branch[1]{[#1]}
\newcommand*\transclos[1]{{#1}^{+}}
\newcommand*\reverse[1]{-(#1)}
\newcommand{\union}{\cup}
\newcommand{\conjunction}{\cap}
\newcommand{\queryTerm}[1]{\phi_{#1}}

\newcommand{\variable}[1]{\mathcal{V}_{#1}}

\newcommand{\PathPatternSet}[1]{\mathcal{P}_{#1}}

\newcommand{\cqt}{\texttt{CQT}}
\newcommand{\ucqt}{\texttt{UCQT}}
\newcommand{\nre}{\texttt{NRE}}
\newcommand{\rpq}{\texttt{RPQ}}

\newcommand{\tworpq}{\texttt{2RPQ}}
\newcommand{\ctworpq}{\texttt{C2RPQ}}
\newcommand{\uctworpq}{\texttt{UC2RPQ}}
\newcommand{\cnre}{\texttt{CNRE}}
\newcommand{\ucnrpq}{\texttt{UCN2RPQ}}

\newcommand{\cq}{\texttt{CQ}}
\newcommand{\ucq}{\texttt{UCQ}}
\newcommand{\gxpath}{\texttt{GXPath}}

\newcommand{\mRA}{\mu\texttt{-RA}}

\newcommand{\srcCol}{\mathsf{Sr}}
\newcommand{\trgCol}{\mathsf{Tr}}

\newcommand{\muTerm}[1]{\varphi_{#1}}

\newcommand{\rename}[3]{\rho_{#1}^{#2}\left(#3\right)}

\newcommand{\proj}[2]{\pi_{#1}\left(#2\right)}

\newcommand{\NJoin}{\bowtie}

\newcommand{\Head}[1]{\mathsf{H}_{#1}}
\newcommand{\Body}[1]{\mathsf{B}_{#1}}
\newcommand{\Atomic}[1]{\mathsf{A}_{#1}}
\newcommand{\CQTRelation}[1]{\mathsf{Rel}_{#1}}

\usepackage{todonotes}

\definecolor{blu}{rgb}{0.0, 0.0, 1.0}
\newcommand{\SIGMODchange}[1]{{\color{black}#1}}

\newcommand{\exQuery}[1]{\mathcal{Q}_{#1}}

\title{Schema-Based Query Optimisation for Graph Databases}

\date{} 					

\author{Chandan Sharma, Pierre Genev\`es, Nils Gesbert, Nabil Laya\"ida\\
Tyrex team, Univ. Grenoble Alpes, CNRS, Inria, Grenoble INP, LIG\\
Grenoble, France, 38000
}

\renewcommand{\headeright}{}
\renewcommand{\undertitle}{}
\renewcommand{\shorttitle}{}

\begin{document}
\maketitle

\begin{abstract}
	Recursive graph queries are increasingly popular for extracting information from interconnected data found in various domains such as social networks, life sciences, 
	and business analytics. Graph data often come with schema information that describe how nodes and edges are organized. We propose a type inference mechanism that enriches 
	recursive graph queries with relevant structural information contained in a graph schema. We show that this schema information can be useful in order to improve the performance when 
	evaluating recursive graph queries. Furthermore, we prove that the proposed method is sound and complete, ensuring that the semantics of the query is preserved during the schema-enrichment process.	
\end{abstract}


\section{Introduction}
\nobreak 

The creation, utilisation and, most importantly, analysis of highly interconnected data has become pervasive in various domains, including social media, astronomy, chemistry, bio-informatics, transportation networks and semantics associations (criminal investigation)~\cite{bell2009beyond,barrett2000formal,sheth2005semantic,yang2020graph}. Graph databases have become appealing for modeling, managing and analysing highly interconnected data~\cite{sharma2021practical,alotaibi2021property}. 

Discovering complex relationships between graph-structured data requires expressive graph query languages to use recursion to navigate paths connecting nodes in a graph database~\cite{jachiet2020optimization,duschka2000recursive}. 
Therefore, the design of contemporary graph query languages is based on formalisms such as~\emph{regular path queries} (\rpq{}),~\emph{two-way regular path queries} (\tworpq{}) and~\emph{nested regular expressions} (\nre{}) 
along with their extensions such as~\emph{conjunctive two-way regular path queries} (\ctworpq{})~\cite{calvanese1999rewriting},~\emph{union of conjunctive two-way regular path queries} (\uctworpq{})~\cite{bonifati2018querying,angles2017foundations},
~\emph{conjunctive nested regular expressions} (\cnre{}),~\emph{union of conjunctive nested two-way regular path queries} (\ucnrpq{})~\cite{barcelo2013schema,libkin2018trial,reutter2017regular}, XPath for graph databases (\gxpath{})~\cite{vrgoc2014querying,libkin2013querying} and more recently~\emph{conjunctive queries and union of conjunctive queries extended with Tarski's algebra} (\cqt{}/\ucqt{})~\cite{sharma2021practical}. 
These foundations form the basis of languages such as SPARQL~\cite{hogan2020sparql}, Cypher~\cite{francis2018cypher} and PGQL~\cite{van2016pgql}.
Furthermore, projects such as ISO/IEC 39075~\footnote{GQL:~\url{https://www.iso.org/standard/76120.html}} and Linked Data Benchmark Council~\footnote{LDBC:~\url{https://ldbcouncil.org}} (LDBC) are working towards creating a standard query language for graph databases. 
At the core of all these graph query languages is~\emph{recursion}.
It is essential to perform complex navigation and data extraction from the graph.

Furthermore, graph data often follow a certain organization, structural patterns, or even sets of constraints on the admissible graph shapes. Such knowledge may be left implicit or made explicit by the means of a so-called \textit{graph schema} specification.  In an ongoing standardization effort, \emph{PG-Schema}~\cite{angles2023pg} and~\emph{PG-Keys}~\cite{angles2021pg} have been proposed to serve as a reference graph schema language for graph databases.

One fundamental idea of the work presented in this paper is to leverage the schema constraints in order to improve query evaluation.
Intuitively, we conjecture that taking advantage of the constraints expressed in a schema can be useful to reduce the amounts of graph data involved when evaluating a query. This research proposes a schema-based query rewriting approach to optimize graph queries. We use a basic yet expressive graph schema formalism (based on~\cite{angles2023pg}) to express graph constraints. We consider graph queries expressed in the formalism of \ucqt{}. 
We propose a type inference method for injecting relevant schema information into the query and produce a schema-aware \ucqt{} variant that is semantically equivalent.  

The optimization of recursive queries (even independently from schema constraints) is already known to be significantly more challenging than in the non-recursive setting~\cite{nguyen2017estimating,yakovets2015waveguide, jachiet2020optimization}. Extensions to classical relational algebra have been proposed to support recursion~\cite{agrawal1988alpha,gomes2015beta,yakovets2015waveguide,leeuwen2022avantgraph}. As a result, many graph database engines use the query optimisation techniques developed for relational databases~\cite{sun2015sqlgraph,shanmugasundaram1999relational,jindal2014vertexica,xirogiannopoulos2017graphgen}. Numerous optimisation techniques have been developed~\cite{meimaris2017extended,neumann2011characteristic,bornea2013building,jachiet2020optimization} to optimise recursive graph queries independently from the schema. 

However, due to the schema optional nature of contemporary graph databases~\cite{sharma2022flasc}, earlier schema-based methods that have been researched in other settings for recursive queries have been largely ignored. Examples of such earlier methods include static query analyses for Datalog~\cite{bonatti2004decidability,calvanese1998decidability,calvanese2005decidable}, semi-structured~\cite{deutsch2001optimization,calvanese1998decidability,florescu1998query} databases and XML~\cite{benedikt2008xpath,lakshmanan2004testing,geerts2005satisfiability,geneves06} databases. While these works are essentially of theoretical nature, they can still bring useful insights for schema-based graph query analyses and optimization.

The main contributions of this research work are:

\paragraph{\underline{Schema-based approach for query rewriting}} We propose a type inference mechanism capable of leveraging structural information of a graph schema in order to rewrite \ucqt{} queries. In particular, the type inference mechanism enriches the edge label-based navigational graph queries with additional information related to node labels. This approach aims at reducing the size of intermediate subquery results in order to improve the overall query performance. The approach automatically optimises recursive graph queries and is capable of eliminating costly transitive closure operations using schema information. The soundness and completeness of type inference ensure that the semantics of the query is preserved during optimization.

\paragraph{\underline{Prototype implementation}} We have developed a prototype implementation to demonstrate the practical application of the schema-based approach. The prototype translates schema-based rewritten queries into graph patterns for graph databases and into recursive relational algebra for relational databases. We demonstrate the efficiency of the approach on graph and relational database systems.

\paragraph{\underline{Experimental evaluation}} In order to empirically evaluate the effectiveness of the approach, we use the LDBC social network benchmark dataset which is a property graph, and the YAGO dataset which is a knowledge graph. We use third-party recursive and non-recursive queries. Experiments show that the proposed schema-based approach is effective, especially for recursive queries. 

\paragraph{\textbf{Organisation:}} We present preliminary concepts related to graph databases in Sec.~\ref{sec:PropertyGraphs}, including schemas and the \ucqt{} graph query language. The main contribution to schema-based query rewriting is presented 
in Sec.~\ref{sec:logicalQueryOpt}, where we present an inference system to enrich queries with annotations deduced from the schema. The prototype implementation is presented in Sec.~\ref{sec:Implementation}, where we describe the 
overall system architecture. We then report on empirical evaluation results in Sec.~\ref{sec:experiments}, where we also consider several (relational and graph-based) database systems.  Finally, we discuss related works in Sec.~\ref{sec:RelatedWork} before concluding in Sec.~\ref{sec:Conclusion}. 

\section{Graph Databases}\label{sec:PropertyGraphs}
\nobreak 

We present a few basic definitions related to graph databases. In particular, we define the notions of graph schema, graph database, schema-database consistency and graph queries. 
These notions are vital for schema-based query rewriting. Furthermore, we use the graph database 
representation of the YAGO dataset as a running example to illustrate various definitions present in subsequent sections.

\subsection{Graph Schema}\label{subsec:PGSechma}
\nobreak

A \emph{graph schema} is a directed pseudo multigraph: a graph in which loops on nodes and multiple directed edges between two nodes are permitted. A graph schema captures the structural and properties-based restrictions in a graph database, with nodes representing entities and edges representing relationships between entities. In graph schemas, nodes and edges are labeled, and properties associated with nodes are expressed in the form of \emph{key-type} pairs. Let $\Label{\Node{}}$ be a finite set of node labels and $\Label{\Edge{}}$ be 
a finite set of edge labels such that $\Label{\Node{}} \cap \Label{\Edge{}} = \emptyset$. 
Let $\Key{\Schema{}}$ be a set of 
keys (for example: id, name, age), and $\Type{}$ be a finite set of data types (for example: String, Integer, Date). We define a finite set of properties 
$\Property{\Schema{}}$ such that $\Property{\Schema{}} \subseteq (\Key{\Schema{}} \times \Type{})$. 

\begin{figure}[htbp]
  \centering
  \small 
  \resizebox{.86\linewidth}{!}{%
   \begin{tikzpicture}%
    [>=stealth,
     shorten >=1pt,
     node distance=4cm,
     on grid,
     auto,
     every state/.style={draw=black!100, thick}
    ]
  \node[state,fill=green!5]   (n10)                 {\begin{tabular}{c}  PERSON \\ name:String\\age:Int \end{tabular}};
  \node[state,fill=blue!15]   (n12) [right=of n10]  {\begin{tabular}{c}  CITY \\ name:String \end{tabular}};
  \node[state, fill=red!5]    (n11) [left=of n10]  {\begin{tabular}{c}  PROPERTY \\ address:String \end{tabular}};
  \node[state,fill=red!15]    (n13) [right=of n12]  {\begin{tabular}{c}  REGION \\ name:String \end{tabular}};
  \node[state,fill=green!15]  (n14) [right=of n13]  {\begin{tabular}{c}  COUNTRY \\ name:String \end{tabular}};
  
  \path[->]
   (n10)          edge[loop above]                  node                   {\begin{tabular}{c}  isMarriedTo \end{tabular}}                     (n10)
   (n10)          edge                              node                   {\begin{tabular}{c}  livesIn  \end{tabular}}                        (n12)
   (n10)          edge                              node[swap]             {\begin{tabular}{c}   owns   \end{tabular}}                         (n11)
   (n11)          edge[bend right=45]               node                   {\begin{tabular}{c}   isLocatedIn  \end{tabular}}                   (n12)
   (n12)          edge[bend right=0]               node[swap]             {\begin{tabular}{c}  isLocatedIn \end{tabular}}                     (n13)
   (n13)          edge[bend right=0]               node[swap]             {\begin{tabular}{c}  isLocatedIn  \end{tabular}}                    (n14)
   (n14)          edge[loop above]                  node                   {\begin{tabular}{c}  dealsWith  \end{tabular}}                      (n14)
     ;
  \end{tikzpicture}
  }
  \caption{Sample graph schema for YAGO dataset.}
  \label{fig:YAGOSchema}
\end{figure}

\begin{definition}[Graph Schema]\label{def:PGSchema}
  A graph schema is a tuple $\Schema{} = (\Node{\Schema{}}, \Edge{\Schema{}}, \Label{\Node{}}, \Label{\Edge{}}, \Property{\Schema{}}, \EdgeSrcTrgLabeling{\Schema{}}, 
\NodeLabeling{\Schema{}},\EdgeLabeling{\Schema{}}, \PropertyLabeling{\Schema{}})$ where,

\begin{itemize}
   \item $\Node{\Schema{}}$ and $\Edge{\Schema{}}$ are finite set of nodes and edges such that $\Node{\Schema{}} \cap \Edge{\Schema{}} = \emptyset$.
   
   \item $\EdgeSrcTrgLabeling{\Schema{}}:\Edge{\Schema{}} \rightarrow \Node{\Schema{}} \times \Node{\Schema{}}$ is a function which maps all edges to source and target nodes.
   
   \item $\NodeLabeling{\Schema{}} : \Node{\Schema{}} \rightarrow \Label{\Node{}}$ is a function which maps all nodes to the set of node labels. 
   
   \item $\EdgeLabeling{\Schema{}} : \Edge{\Schema{}} \rightarrow \Label{\Edge{}}$ is a function which maps all edges to the set of edge labels. 
   
   \item $\PropertyLabeling{\Schema{}} : \Node{\Schema{}}\rightarrow 2^{\Property{\Schema{}}}$ is a function which maps all nodes to all possible subsets of the property set $\Property{\Schema{}}$.  
\end{itemize}
\end{definition}

\begin{example}\label{examp:PGSchema}
 Fig.~\ref{fig:YAGOSchema} shows a graph schema for the YAGO dataset \cite{DBLP:journals/ai/HoffartSBW13} consisting of five nodes and seven edges. All edges are labeled and directed; some edges can represent loops on the same nodes. 
 For instance, the edge labeled as \texttt{isMarriedTo} has the same source and target node labeled as \texttt{PERSON}. All nodes are labeled and have properties associated with them. For instance, \SIGMODchange{a node is labeled as \texttt{REGION}} 
 and has a property \texttt{name:String} as a \emph{key:type} pair. 
 
 \end{example}

\subsection{Graph Database}\label{subsec:PGDatabase}
\nobreak

A graph database is a graph instance that allows modeling real-world entities as labeled nodes and edges, with properties associated with nodes. Let $\Key{\Database{}}$ be an infinite set of 
keys (for example, id, name, age), $\Value{}$ be a finite set of values (for example, 345, \texttt{James}). We define a function $\ValueToType{} : \Value{} \rightarrow \Type{}$ that 
maps values in $\Value{}$ to their respective data types in $\Type{}$. The set of properties associated with the nodes of a graph database are defined as $\Property{\Database{}}$ such that $\Property{\Database{}} \subseteq (\Key{\Database{}} \times \Value{})$ where each $p_d \in \Property{\Database{}}$
is a \emph{key-value} pair and each value has a data type. Formally, a graph database is defined as follows:

\begin{definition}[Graph database]\label{def:PGDatabase}
  A graph database is a tuple $\Database{} = (\Node{\Database{}}, \Edge{\Database{}}, \Label{\Node{}}, \Label{\Edge{}}, 
  \Property{\Database{}}, \EdgeSrcTrgLabeling{\Database{}}, 
  \NodeLabeling{\Database{}},\EdgeLabeling{\Database{}}, \PropertyLabeling{\Database{}})$ where,
  \begin{itemize}
      \item $\Node{\Database{}}$ and $\Edge{\Database{}}$ are finite set of nodes and edges such that $\Node{\Database{}} \cap \Edge{\Database{}} = \emptyset$. 
      \item $\EdgeSrcTrgLabeling{\Database{}}:\Edge{\Database{}} \rightarrow \Node{\Database{}} \times \Node{\Database{}}$ is a function which maps all edges to source and target nodes.
      \item $\NodeLabeling{\Database{}} : \Node{\Database{}} \rightarrow \Label{\Node{}}$ is a function which maps all nodes to the set of node labels. 
      \item $\EdgeLabeling{\Database{}} : \Edge{\Database{}} \rightarrow \Label{\Edge{}}$ is a function which maps all edges to the set of edge labels. 
      \item $\PropertyLabeling{\Database{}} : \Node{\Database{}} \rightarrow 2^{\Property{\Database{}}}$ is a function which maps all nodes to all possible subsets of the property set $\Property{\Database{}}$.
  \end{itemize}
\end{definition}

\begin{figure}[htbp]
  \tiny
  \centering
  \resizebox{.7\linewidth}{!}{%
   \begin{tikzpicture}%
    [>=stealth,
     shorten >=1pt,
     node distance=2.5cm,
     on grid,
     auto,
     every state/.style={draw=black!100, thick}
    ]
  \node[state,fill=red!5, scale=.8]   (n1)                     {\begin{tabular}{c} $n_1$\\ PROPERTY \\ address:\\7 Queen \\ Street \end{tabular}};
  \node[state, fill=green!5, scale=.8]    (n2) [right=of n1]      {\begin{tabular}{c} $n_2$\\PERSON \\ name:John \\ age:28 \end{tabular}};
  \node[state, fill=green!5, scale=.8]    (n3) [right=of n2]      {\begin{tabular}{c} $n_3$\\PERSON \\ name:\\Shradha \\ age:25 \end{tabular}};
  \node[state, fill=blue!15, scale=.8]    (n4) [right=of n3]    {\begin{tabular}{c} $n_4$\\CITY \\ name:\\Elerslie \end{tabular}};
  \node[state,fill=red!15, scale=.8]    (n5) [below=of n2]      {\begin{tabular}{c} $n_5$\\REGION \\ name:\\ Grenoble \end{tabular}};
  \node[state,fill=blue!15, scale=.8]   (n6) [below=of n1]      {\begin{tabular}{c} $n_6$\\CITY \\ name:\\ Montbonnot \end{tabular}};
  \node[state,fill=green!15, scale=.8]   (n7) [below=of n4]      {\begin{tabular}{c} $n_7$\\COUNTRY \\ name:\\ France \end{tabular}};
  
  \path[->]
      (n2)           edge[bend left=18]                                     node                        {\begin{tabular}{c}  isMarriedTo \end{tabular}}                            (n3)
      (n3)           edge[bend left=18]                       node                        {\begin{tabular}{c}  isMarriedTo  \end{tabular}}                            (n2)
      (n3)           edge                                     node                        {\begin{tabular}{c}  livesIn  \end{tabular}}                                (n4)
      (n2)           edge[bend right=45]                       node                  {\begin{tabular}{c}  livesIn \end{tabular}}                                (n4)
      (n2)           edge                                     node[swap]                  {\begin{tabular}{c}  owns \end{tabular}}                                  (n1)
      (n4)           edge[bend left=25]                       node[swap]                        {\begin{tabular}{c}  isLocatedIn  \end{tabular}}                                               (n5)
      (n1)           edge                                     node                        {\begin{tabular}{c}  isLocatedIn \end{tabular}}                                                (n6)
      (n6)           edge[bend right=20]                                     node[swap]                  {\begin{tabular}{c}  isLocatedIn \end{tabular}}                                                (n5)
      (n5)           edge[bend right=10]                                     node[swap]                  {\begin{tabular}{c}  isLocatedIn \end{tabular}}                                                (n7)    
     ;
  \end{tikzpicture}
  }
  \caption{An example of YAGO graph database.}
  \label{fig:YAGODatabase}
\end{figure}

\begin{example}\label{examp:YAGODatabase}
    Fig.~\ref{fig:YAGODatabase} shows an example of a YAGO graph database consisting of seven nodes and nine edges. All nodes are labeled, and have optional properties associated with them. All edges are labeled and can be identified by unique source and target node identifiers using function $\EdgeSrcTrgLabeling{\Database{}}$. For instance, edge $(n_2, n_1)$ is labeled as \texttt{owns}, with $n_2$ as the source and $n_1$ as the target node identifier. The node with identifier $n_{2}$ is labeled as \texttt{PERSON} and has \texttt{name:John} and \texttt{age:28} as properties. 
    
\end{example}

\subsection{Schema-database consistency}\label{subsec:SDConsistency}
\nobreak 
The notion of schema-database consistency implies that a graph database follows the structural and property based restrictions established by a graph schema. We define a  function $\SchemaDatabaseMap{}: \Database{} \rightarrow \Schema{}$ that maps elements in the graph database to at most one element in the graph schema (also known as strict graph schema~\cite{angles2023pg}).

\begin{definition}[Schema database consistency]\label{def:SDConsistency}
  Given a graph database $\Database{}$ and a graph schema $\Schema{}$,
  we say that $\Database{}$ is
  \emph{consistent} with $\Schema{}$ when there exists a schema-database
  mapping $\SchemaDatabaseMap{}$ such that:
    \begin{itemize}
        \item For every node $n_i \in \Database{}$ there exists a
          corresponding node in the schema:
          $\SchemaDatabaseMap{}(n_i) \in \Schema{}$, and we have $\NodeLabeling{\Database{}}(n_i) = \NodeLabeling{\Schema{}}(\SchemaDatabaseMap{}(n_i))$.      
        
        \item For every edge $e_i \in \Database{}$ there exists a
          corresponding edge in the schema: $
          \SchemaDatabaseMap{}(e_i) \in \Schema{}$. Let
          $(n_i,n_j) = \EdgeSrcTrgLabeling{\Database{}}(e_i)$, then we have
        $\EdgeSrcTrgLabeling{\Schema{}}(\SchemaDatabaseMap{}(e_i)) =
        (\SchemaDatabaseMap{}(n_i), \SchemaDatabaseMap{}(n_j))$;
        furthermore, we have $\NodeLabeling{\Database{}}(n_i) = \NodeLabeling{\Schema{}}(\SchemaDatabaseMap{}(n_i))$,
        $\NodeLabeling{\Database{}}(n_j) = \NodeLabeling{\Schema{}}(\SchemaDatabaseMap{}(n_j))$ and $\EdgeLabeling{\Database{}}(e_i) = \EdgeLabeling{\Schema{}}(\SchemaDatabaseMap{}(e_i))$.
        
        \item For each $n_i \in \Node{\Database{}}$ and for each $(k,
          v) \in \PropertyLabeling{\Database{}}(n_i)$, we have $(k, 
          \ValueToType{}(v))\in
          \PropertyLabeling{\Schema{}}(\SchemaDatabaseMap{}(n_i))$.
    \end{itemize}
\end{definition}

\begin{example}\label{examp:SIConsistency}
    By using Def.~\ref{def:SDConsistency}, we can observe that the graph database presented in Fig.~\ref{fig:YAGODatabase} is consistent with the graph schema shown in Fig.~\ref{fig:YAGOSchema}. For instance, edge ($n_2, n_3$) is labeled as \texttt{isMarriedTo}; moreover, source and target nodes are labeled as \texttt{PERSON} and follow the property-based restrictions imposed in the YAGO graph schema.  
    
\end{example}

In this study, we only consider graph databases that are consistent with the graph schema; that is, we exclude any database that does not conform to Def.~\ref{def:SDConsistency}. 
Furthermore, we consider graph databases with restrictions, including \SIGMODchange{each node/edge can have at most one node/edge label associated with it, and edges do not have properties associated with them. However, as mentioned in~\cite{vrgoc2014querying} 
our graph data model can easily be modified to accommodate properties over edges. In our graph data model we allow zero or more properties associated with nodes.}
Moreover, we consider that properties are atomic entities and cannot have maps nor lists as data types\footnote{Interested readers may refer to~\cite{angles2020mapping,sharma2021practical,sharma2022flasc} for detailed discussions on graph data model restrictions.}.

\subsection{Querying graph databases}\label{sec:queryingPG}
\nobreak 

\emph{Graph patterns} are essential in graph query languages as they assist in defining the structure of data to be extracted from a graph database~\cite{angles2017foundations,francis2023gpc,sharma2021practical}. 

\label{subsec:GraphPattern}

The core component of a graph pattern is a~\emph{path expression} corresponding to specifying directed edges and/or paths defined over the edge labels of a graph database. 
To specify path expressions we use the formalism of~\emph{Tarski's algebra} which is strictly more expressive than other graph query formalisms 
such as~\emph{two-way regular path queries} ($\tworpq$) and~\emph{nested regular expressions} ($\nre$)~\cite{hellings2018tarski,hellings2017relation}. 
The grammar of Tarski's algebra used to formulate a path expression $\queryTerm{}$  is presented in Fig.~\ref{fig:Tarski}. 

\begin{figure}[htbp] 
    \footnotesize
    \begin{align*} 
        \queryTerm{} ::= && && \text{path expression} \\ 
         && \SIGMODchange{l_e}  && \SIGMODchange{\text{single edge label } (l_e \in \Label{\Edge{}})}  \\
        | && \queryTerm{1} \concatenation \queryTerm{2} && \text{concatenation} \\
        | && \queryTerm{1} \union \queryTerm{2} && \text{union} \\
        |&&  \queryTerm{1} \conjunction \queryTerm{2} && \text{conjunction} \\ 
        |&&  \queryTerm{1} [\queryTerm{2}] && \text{branch (right)}\\
        | && [\queryTerm{1}] \queryTerm{2} && \text{branch (left)} \\
        | && -l_e && \text{reverse}\\
        | && \queryTerm{}^{+} && \text{transitive closure}
    \end{align*}
     
    \caption{Tarski's algebra grammar (adapted from~\cite{hellings2018tarski,sharma2021practical}).}
    \label{fig:Tarski}
\end{figure}

In our adaptation of Tarski's algebra, the~\emph{reverse operation} can be used with single-edge labels, as shown in Fig.~\ref{fig:Tarski}. Notice that it is possible to use the reverse operator 
in front of general path expressions $\queryTerm{}$, as this does not offer any additional expressive power compared to the reverse operation on single-edge labels~\cite{hellings2018tarski}. 

\SIGMODchange{In the sequel, we define and use the formalism of \emph{conjunctive queries and union of conjunctive queries extended with Tarski's algebra} (\cqt/\ucqt) to express graph patterns based on path expressions. We choose the \ucqt{} query language formalism for its high expressive power as it encompasses most existing query language formalisms. In particular, \ucqt{} subsumes both \emph{union of conjunctive two-way regular path queries} (\uctworpq{}) 
and~\emph{union of conjunctive nested two-way regular path queries} (\ucnrpq{}) which are used in many practical graph query languages such as Cypher, SPARQL and PGQL~\cite{sharma2021practical,vrgoc2014querying,reutter2017regular}.

Additionally, in \ucqt{}, the conjunction operator is closed under transitive closure~\cite{sharma2021practical,hellings2018tarski}, which adds the ability to express recursive queries that search for 
arbitrary-length paths with repeating cyclic structures between the start and end nodes. Such queries cannot be expressed in the formalisms of \uctworpq{} and \ucnrpq{}~\cite{sharma2021practical}.

}

\subsubsection{Syntax of $\cqt/\ucqt$}\label{subsubsec:cqtSyntax}
\nobreak 

Given a graph schema $\Schema{}$ and a graph database $\Database{}$, let $\variable{\Node{}}$ be a finite set of node variables and $\PathPatternSet{\queryTerm{}}$ be a set of path expressions $\{\queryTerm{1},\ldots,\queryTerm{n}\}$, such that each path expression $\queryTerm{i} \in \PathPatternSet{\queryTerm{}}$ is defined over the set of edge labels $\Label{\Edge{}}$. 

\begin{definition}[$\cqt$]\label{def:cqt}
    A conjunctive query with Tarski's algebra is a logical formula in the $\exists, \wedge$-fragment of first order logic, of the form $\CQTQuery{}$ = \big\{($h_1,\ldots,h_i$)  $\vert \hspace{.4em} \exists$ $(b_1,$ $\ldots,b_j) \hspace{.4em} r_1 \wedge \ldots \wedge r_l \wedge \hspace{.4em} a_1 \wedge \ldots \wedge a_k$  \big\} where,

    \begin{itemize}
        \item $\Head{} =$ \SIGMODchange{\{$h_1,\ldots,h_i$\}} is a finite set of head variables. $\Body{} =$ \SIGMODchange{\{$b_1,\ldots,b_j$\}} is a finite set of body variables such that $(\Body{} \cup \Head{}) \subseteq \variable{\Node{}}$ and \SIGMODchange{$(\Body{} \cap \Head{}) = \emptyset$.} 
        \item $\Atomic{} = $\SIGMODchange{ \{$a_1,\ldots,a_k$\}} is a finite set of atomic formulas formed by a labeling function $\NodeLabeling{\Atomic{}}: \variable{\Node{}} \rightarrow (\Label{\Node{}} \cup \epsilon)$ that maps all node variables
        to node labels in $\Label{\Node{}}$ or to the empty label. These formulas can be of the form \emph{e.\ g.}\ $\NodeLabeling{\Atomic{}}(Y) = \texttt{PERSON}$ to specify that nodes represented by the variable $Y$ must be labeled as \texttt{PERSON}.
        \item $\CQTRelation{} \subseteq (\variable{\Node{}} \times \PathPatternSet{\queryTerm{}} \times \variable{\Node{}})$ is a finite set of relations where each relation $r_i \in \CQTRelation{}$ either represent directed edge(s) or path(s) connecting two nodes.
    \end{itemize}
\end{definition}

\begin{figure}[htbp]
    \centering
    \tiny
    \resizebox{.4\linewidth}{!}{%
     \begin{tikzpicture}%
      [>=stealth,
       shorten >=1pt,
       node distance=2cm,
       on grid,
       auto,
       every state/.style={draw=black!100, very thick}
      ]
    \node[state,fill=red!5 ]       (z)              {\begin{tabular}{c} Z  \end{tabular}};  
    \node[state, fill=green!5]     (y) [right=of z] {\begin{tabular}{c} Y  \end{tabular}};
    \node[state, fill=green!15]    (m) [right=of y] {\begin{tabular}{c} M \end{tabular}};

    \path[->]
       (y)         edge                node[swap]   {\begin{tabular}{c} owns \end{tabular}}    (z)
       (y)         edge[bend left=30]                node         {\begin{tabular}{c} livesIn$\concatenation$isLocatedIn$^{+}$\end{tabular}}    (m)       
       ;
    \end{tikzpicture}
    }
     
    \caption{Example of a graph pattern.}
    \label{fig:graphPattern}
\end{figure}

\begin{example}
  The graph pattern in Fig.~\ref{fig:graphPattern} identifies people living in regions and countries who also own properties. 
  It consists of two relations:  (\texttt{Y}, \texttt{owns}, \texttt{Z}) searches for people who own properties. While (\texttt{Y}, \texttt{livesIn$\concatenation$isLocatedIn$^{+}$}, \texttt{M}) 
  specifies a path expression to \SIGMODchange{search} for paths that have edges labeled as \texttt{livesIn} followed by an unbounded number of edges labeled as \texttt{isLocatedIn} returning node that either correspond to regions or countries. 
  Both relations share the same node variable \texttt{Y}, specifying that both relations must search for the same person.        
\end{example}

\small 
\begin{align*}
    \CQTQuery{1} = \{\texttt{Y} \hspace{.7em} \vert \hspace{.7em} \exists (\texttt{Z,M}) \hspace{.7em} (\texttt{Y,(livesIn\concatenation{}isLocatedIn+)},\texttt{M}) \wedge (\texttt{Y,owns,Z}) \}
\end{align*}

\normalsize

\begin{example}
    Query $\CQTQuery{1}$, presents a graph pattern in Fig.~\ref{fig:graphPattern} expressed as a $\cqt$ where \texttt{Z,M} are body variables and \texttt{Y} is a head variable. Relations $(\texttt{Y,owns,Z})$ and $(\texttt{Y,livesIn\concatenation{}isLocatedIn}^{+},\texttt{M})$ describe the structure of the graph pattern. 
\end{example}

We extend the $\cqt$ query language formalism to the formalism of the union of conjunctive queries with Tarski's algebra ($\ucqt$) which is analogous to 
extending the formalism of $\ctworpq$ to $\uctworpq$. The formalism of $\ucqt$ represents the disjunction of conjunctive queries with 
Tarski's algebra. A $\ucqt$ query is written as $\UCQTQuery{} = \Head{}\leftarrow\{\CQTQuery{1}\union \ldots \union \CQTQuery{n}\}$, where all $\CQTQuery{i}$ are \emph{union compatible} $\cqt$ queries, that is, they share the same set of head variables $\Head{}$~\cite{perez2006semantics,jachiet2020optimization}.

\subsubsection{Semantics of $\cqt/\ucqt$}\label{subsubsec:cqtSyntax}
\nobreak

As discussed in Sec.~\ref{subsec:GraphPattern}, path expressions form the core components for syntactically describing graph patterns as $\cqt/\ucqt$. We first present the semantics of path expressions $\queryTerm{}$ that are based on Tarski's algebra~\footnote{Interested readers may refer~\cite{sharma2021practical} for detailed discussion on the semantics of $\cqt/\ucqt$}. 
The interpretation of the path expression $\queryTerm{}$ when evaluated over a graph database $\Database{}$ (represented as $\llbracket \queryTerm{}\rrbracket_{\Database{}}$) is defined in Fig.~\ref{fig:intercqtGrammar}. 
The output produced after evaluating a path expression $\queryTerm{}$ consists of all pairs of nodes (source and target nodes) that are connected by the path $\queryTerm{}$ in the graph database.

\begin{figure}[htbp]
    \centering
    \footnotesize
    \begin{align*}
        &&\llbracket l_e \rrbracket_{\Database{}} =&&  \{(n,m) | n \xrightarrow{l_e} m \in \Edge{\Database{}} 
        \wedge n,m \in \Node{\Database{}} \wedge l_e \in \Label{\Edge{}} \}\\
        &&\llbracket \queryTerm{1} \concatenation \queryTerm{2} \rrbracket_{\Database{}} =&& \{(n,m) | \exists z \in \Node{\Database{}} 
        \wedge (n,z) \in \llbracket \queryTerm{1} \rrbracket_{\Database{}} 
        \wedge (z,m) \in \llbracket \queryTerm{2} \rrbracket_{\Database{}}\} \\
        &&\llbracket \queryTerm{1} \union \queryTerm{2} \rrbracket_{\Database{}} =&& \llbracket \queryTerm{1} \rrbracket_{\Database{}}
        \union \llbracket \queryTerm{2} \rrbracket_{\Database{}} \\
        &&\llbracket \queryTerm{1} \conjunction \queryTerm{2} \rrbracket_{\Database{}} =&& \llbracket \queryTerm{1} \rrbracket_{\Database{}}
        \conjunction \llbracket \queryTerm{2} \rrbracket_{\Database{}} \\
        &&\llbracket \queryTerm{1}[\queryTerm{2}] \rrbracket_{\Database{}} =&& \{(n,m) | \exists z \in \Node{\Database{}} 
        \wedge (n,m) \in \llbracket \queryTerm{1} \rrbracket_{\Database{}}
        \wedge (m,z) \in \llbracket \queryTerm{2} \rrbracket_{\Database{}} \}\\
        &&\llbracket [\queryTerm{1}]\queryTerm{2} \rrbracket_{\Database{}} =&& \{(n,m) | \exists z \in \Node{\Database{}} 
        \wedge (n,z) \in \llbracket \queryTerm{1} \rrbracket_{\Database{}}
        \wedge (n,m) \in \llbracket \queryTerm{2} \rrbracket_{\Database{}} \} \\
        &&\llbracket -l_e \rrbracket_{\Database{}} =&& \{(m,n) | (n,m) \in \llbracket l_e \rrbracket_{\Database{}}\}\\
        &&\llbracket \queryTerm{}^{+} \rrbracket_{\Database{}} =&& \bigcup \limits_{i \geq 1}^{} \llbracket \queryTerm{}^{i} \rrbracket_{\Database{}} 
        , \queryTerm{}^{k} = (\underbrace{\queryTerm{} \concatenation \ldots \concatenation \queryTerm{}}_{k-\text{times}}) \text{ where } 1 \leq \SIGMODchange{k} \leq i.  
    \end{align*}
     
    \caption{Semantics of Tarski's algebra (adapted from~\cite{sharma2021practical}).}
    \label{fig:intercqtGrammar}
\end{figure}

\SIGMODchange{
 \begin{example}\label{example:PathExpression}
    Given a graph schema (Fig.~\ref{fig:YAGOSchema}) and a graph database (Fig.~\ref{fig:YAGODatabase}) for the YAGO dataset, consider $\queryTerm{1} = [\texttt{owns}]([\texttt{isMarriedTo}]\texttt{livesIn})$ a 
    path expression used to search for all married property owners living in cities. The path expression $\queryTerm{1}$ first searches for people living in cities using the $\texttt{livesIn}$ edge label. 
    After that, the branching operation on $\texttt{isMarriedTo}$ edge label only serves as an existential node test to select married people. Finally, the branching operation on the $\texttt{owns}$ edge 
    label only selects married people living in cities and owning properties. 
    
    Based on the semantics in Fig.~\ref{fig:intercqtGrammar}, the query returns nodes \{($n_2, n_4$)\} as the final result set, where node $n_2$ represents a person ``\texttt{John}'' who lives in a city named ``\texttt{Elerslie}'' (represented by node $n_4$). Furthermore, \texttt{John} is married and owns a property as shown in Fig.~\ref{fig:YAGODatabase}.   
   \end{example}
} 
 The formalism of $\cqt/\ucqt$ expresses queries of two types:~\emph{non-recursive graph queries (NQ)} and~\emph{recursive graph queries (RQ)}. 
 Non-recursive graph queries are restricted $\cqt/\ucqt$ that do not allow transitive closure, whereas recursive graph queries allow general path expressions with transitive closure. 
 
 The semantics of graph query language formalisms are broadly of two types~\emph{(i)} evaluation semantics and \emph{(ii)} output semantics~\cite{angles2017foundations,sharma2021practical}. The formalism of $\cqt/\ucqt$ uses homomorphism-based evaluation semantics for non-recursive graph queries (NQ), arbitrary path semantics for recursive graph queries (RQ) and set-based output semantics~\cite{sharma2021practical}.

\section{Schema-Based Query Rewriting}\label{sec:logicalQueryOpt}
\nobreak 

We introduce a method that leverages schema information for query evaluation. Specifically the method rewrites a query so that structural schema information is injected in relevant part of the queries, namely in path expressions. As a preliminary step, we first transform path expressions into a form where redundancies are eliminated. We then present how schema information can be injected into path expressions.

\paragraph{\underline{Preliminary path simplifications}} 
The purpose of rewrite rules presented in Fig.~\ref{fig:simpleRWrules} is to eliminate redundancies from path expressions in order to simplify queries. These rewrite rules are general in that they apply 
independently from any particular schema.

\begin{figure}[htbp]
  \centering
  \footnotesize
  \begin{align*}
      (\queryTerm{}^{+})^{+} \rightarrow \queryTerm{}^{+} \text{ \texttt{(R1)}} && \queryTerm{1}^{+} [\queryTerm{2}^{+}] \rightarrow \queryTerm{1}^{+}[\queryTerm{2}] \text{ \texttt{(R2)}}  && \queryTerm{1}[\queryTerm{2} \concatenation \queryTerm{3}] \rightarrow \queryTerm{1}[\queryTerm{2}[\queryTerm{3}]] \text{ \texttt{(R3)}}\\
      &&  [\queryTerm{2}^{+}] \queryTerm{1}^{+} \rightarrow [\queryTerm{2}]\queryTerm{1}^{+} \text{ \texttt{(R4)}} && [\queryTerm{2} \concatenation \queryTerm{3}]\queryTerm{1} \rightarrow [\queryTerm{2}[\queryTerm{3}]]\queryTerm{1} \text{ \texttt{(R5)}} 
  \end{align*}  
   
  \caption{Rewrite rules for path expression simplification.}
  \label{fig:simpleRWrules}
\end{figure}

Rule \texttt{R1} removes the redundant use of transitive closures. 

Rules \texttt{R2} and \texttt{R4} simplify transitive closures within the branching operator since computing full transitive closure in branching expressions is not required~\cite{hellings2018tarski,hellings2017relation}. 
Rules \texttt{R3} and \texttt{R5} turn path compositions into branching operations when possible, as also done in prior works \cite{hellings2017relation,hellings2018tarski,hellings2020comparing,hellings2023expressive,jaakkola2023complexity}. 
Correctness of these rules follows immediately from the formal semantics of path expressions, and in particular from the existential semantics of the branching operator defined in Figure~\ref{fig:intercqtGrammar}.

\begin{figure}[htbp]
    \centering
    \footnotesize 
    \begin{align*}
      \queryTerm{red} =     && \texttt{(((owns[isMarriedTo}^{+}\texttt{/livesIn/dealsWith}^{+}\texttt{])/(isLocatedIn}^{+}\texttt{)}^{+}\texttt{)}^{+}\texttt{)}^{+}\\
      \queryTerm{opt} =     && \texttt{((owns[isMarriedTo[livesIn[dealsWith]]]/isLocatedIn}^{+}\texttt{)}^{+}
    \end{align*}
     
    \caption{Application of path simplification rules.}
    \label{fig:simpleRW}
\end{figure}

\begin{example}
    In Fig~\ref{fig:simpleRW}, $\queryTerm{red}$ is a path expression with redundant plus operations and uses the concatenation operation within the branched path expressions. Rule \texttt{R1} removes redundant plus operations, rule \texttt{R2} removes plus operation inside branched expressions, and rule \texttt{R3} replaces concatenation with branching operation in branched expressions. The optimised path expression is presented in Fig.~\ref{fig:simpleRW} as the path expression $\queryTerm{opt}$. 
\end{example}

\subsection{Schema-based query analysis}
\nobreak 

Path expressions in Tarski’s algebra do not contain any information about node labels. However, given a graph schema, it is possible to use the structure of a path expression to infer information about node labels. 
\SIGMODchange{This information can then be used to rewrite the original query into a more precise query. The rewritten query generates fewer intermediate results in general, but is always equivalent to the original 
one in terms of the final result set, when the queried database conforms strictly to the schema.}

\begin{example}
  \SIGMODchange{In the graph database of the YAGO dataset shown in Fig.~\ref{fig:YAGODatabase}, suppose we need to search for a path such as \texttt{owns/isLocatedIn}. 
  To answer this query, we only need to look for paths that start from a node labeled \texttt{PERSON} with an outgoing edge labeled \texttt{owns} to a node labeled \texttt{PROPERTY}, 
  followed by an outgoing edge labeled \texttt{isLocatedIn} to nodes labeled \texttt{CITY}. This path traversal information can be inferred from the graph schema to formulate a precise path traversal query. 
  When we do not use this schema information, the query considers all edges labeled \texttt{isLocatedIn} that start from nodes labeled \texttt{PROPERTY}, \texttt{CITY}, and \texttt{REGION}, respectively.}
\end{example}

\subsubsection{Annotated path expressions}

We first extend the grammar of Tarski’s algebra to annotate
concatenations with node labels: \emph{annotated path expressions}
$\annquery{}$ follow the same grammar as described in
Fig.~\ref{fig:Tarski} except that concatenation $\concatenation$ can be
replaced by its annotated version $\annconcat{\nodelabelvar}$ where
$\nodelabelvar$ is a node label. The expression
$\annquery{1}\annconcat\nodelabelvar\annquery{2}$ represents paths
which follow $\annquery{1}$, arrive at a node labeled
$\nodelabelvar$, and go on from there following $\annquery{2}$.
Formally, we define: 
\small 
\begin{align*}\llbracket
\annquery{1}\annconcat\nodelabelvar\annquery{2}\rrbracket_{\Database{}}
= \{(n,m) \mid~ &\exists z \in
\Node{\Database{}}~\NodeLabeling{\Database{}}(z) = \nodelabelvar
        \wedge (n,z) \in \llbracket \annquery{1} \rrbracket_{\Database{}} 
        \wedge (z,m) \in \llbracket \annquery{2}
        \rrbracket_{\Database{}}\}
\end{align*}
\normalsize
where $\llbracket\cdot\rrbracket$ is defined as in
Fig.~\ref{fig:intercqtGrammar} for the other cases.

Our idea here is that, when querying a database, replacing a plain
path expression with a set of annotated ones can help reduce the size
of intermediary results by keeping only the relevant data and thus
improve efficiency. To determine how these annotations can be added,
we use the schema.

\subsubsection{Graph schema triples}
\nobreak 

Given a graph schema $\Schema{}$ as in
Def.~\ref{def:PGSchema}, we have that for each edge
$e_i \in \Edge{\Schema{}}$, there exists a pair of source and target
nodes $\EdgeSrcTrgLabeling{\Schema{}}(e_i) = (n_i, n_j)$. For each
such edge, we consider the \emph{basic graph schema triple}
$\triplevar_i = \SingleTriple{\nodelabelvar}{\edgelabelvar}{\nodelabelvar'}$
constituted by the source label, the edge label and the target label,
without any information about properties. Formally:

\begin{definition}\label{def:BGSTriples}
  The set of~\emph{basic graph schema triples} of $\Schema{}$,
  $\SchemaTriple{b}(\Schema{})$, is defined as follows:
  \small 
  \begin{align*}\SchemaTriple{b}(\Schema{}) = \{\SingleTriple{\nodelabelvar}{\edgelabelvar}{\nodelabelvar'})\mid~ &\exists n_i, n_j\in\nodesch~~\exists
  e_i\in\edgesch ~ \edgelab(e_i) = \edgelabelvar \wedge ~ \edgesrctrg(e_i)
  = (n_i, n_j) \wedge \nodelab(n_i)= \nodelabelvar \wedge
  \nodelab(n_j) = \nodelabelvar'\}\end{align*}
  \normalsize
\end{definition}

\begin{example}\label{examp:BGSTriple}
    The graph schema as shown in Fig.~\ref{fig:YAGOSchema} contains seven edges; therefore, the set associated with the schema contains seven basic graph schema triples $\SchemaTriple{b}(\Schema{}) = \{t_1,\ldots,t_7\}$. 
    For instance, the triple $t_1$ = (\texttt{PERSON}, \texttt{owns}, \texttt{PROPERTY}) has \texttt{PERSON}  as source node label, \texttt{PROPERTY}  as target node label and \texttt{owns} as an associated edge label. 
    Similarly, the triple $t_2$ = (\texttt{PROPERTY}, \texttt{isLocatedIn}, \texttt{CITY}) has \texttt{PROPERTY}  as source node label, \texttt{CITY}  as target node label and \texttt{isLocatedIn} as an associated edge label.

    Consider the path expression \texttt{owns}. The only triple
    containing this label is $t_1$. If we query a database conforming
    to our schema, we are able to know that this path expression will
    only return results conforming to $t_1$, in the sense that their
    source node will be labeled \texttt{PERSON} and their target node
    \texttt{PROPERTY}.
  \end{example}

Basic graph schema triples correspond to path expressions consisting
of a single edge label. More generally, we define graph schema triples
as follows:
  
\begin{definition}\label{def:GSTriple}
  A \emph{graph schema triple} is a triple $\SingleTriple{\nodelabelvar}
  {\annquery{}} {\nodelabelvar'}$ where $\nodelabelvar$ and
  $\nodelabelvar'$ are node labels and $\annquery{}$ an annotated path
  expression. For a graph schema triple $\triplevar$, we write
  $\tripleSrc{}(\triplevar), \tripleEdgeLabel{}(\triplevar)$ and
  $\tripleTrg{}(\triplevar)$ respectively the source node label,
  annotated path expression and target node label.
\end{definition}

Given a plain path expression and a graph schema, we can compute a
number of graph schema triples which are \emph{compatible} with the
path expression. For example, if we consider the path expression
$\texttt{owns}\concatenation\texttt{isLocatedIn}$, considering that
triples $t_1$ and $t_2$ from
Exp.~\ref{examp:BGSTriple} share a common node label
\texttt{PROPERTY}, we can build the triple
$\SingleTriple{\texttt{PERSON}}{\texttt{owns}\annconcat{\texttt{PROPERTY}}\texttt{isLocatedIn}}{\texttt{CITY}}$.
  In that example, it is the only triple compatible with the expression.

\subsubsection{Path expression and triple compatibility}
\nobreak 

\begin{figure*} 
  \scriptsize  
    \begin{align*}
      \inference
      {\SingleTriple{\nodelabelvar}{\edgelabelvar}{\nodelabelvar'}\in\BasicSchemaTriples(\Schema{})}
      {\judg{\edgelabelvar}{\SingleTriple{\nodelabelvar}{\edgelabelvar}{\nodelabelvar'}}}
      &\textsc{(TBasic)}&
      \inference
      {\judg{\queryTerm{1}}{\SingleTriple{\nodelabelvar}{\annquery{1}}{\nodelabelvar'}} &
       \judg{\queryTerm{2}}{\SingleTriple{\nodelabelvar'}{\annquery{2}}{\nodelabelvar''}}}
      {\judg{\queryTerm{1}\concatenation\queryTerm{2}}{\SingleTriple{\nodelabelvar}{\annquery{1}\annconcat{\nodelabelvar'}\annquery{2}}{\nodelabelvar''}}}
      &\textsc{(TConcat)}&
      \inference
      {\judg{\queryTerm{}}{\SingleTriple{\nodelabelvar}{\annquery{}}{\nodelabelvar'}}}
      {\judg{\reverse{\queryTerm{}}}{\SingleTriple{\nodelabelvar'}{\reverse{\annquery{}}}{\nodelabelvar}}}
      &\textsc{(TMinus)}
      \\
      \inference
      {\judg{\queryTerm{2}}{\triplevar}}
      {\judg{\queryTerm{1}\union\queryTerm{2}}{\triplevar}}
      &\textsc{(TUnionR)}&
      \inference
      {\judg{\queryTerm{1}}{\SingleTriple{\nodelabelvar}{\annquery{1}}{\nodelabelvar'}} &
       \judg{\queryTerm{2}}{\SingleTriple{\nodelabelvar'}{\annquery{2}}{\nodelabelvar''}}}
      {\judg{\queryTerm{1}\branch{\queryTerm{2}}}{\SingleTriple{\nodelabelvar}{\annquery{1}\branch{\annquery{2}}}{\nodelabelvar'}}}
      &\textsc{(TBranchR)}&
      \inference
      {\judg{\queryTerm{1}}{\SingleTriple{\nodelabelvar}{\annquery{1}}{\nodelabelvar'}} &
       \judg{\queryTerm{2}}{\SingleTriple{\nodelabelvar}{\annquery{2}}{\nodelabelvar'}}}
      {\judg{\queryTerm{1}\conjunction\queryTerm{2}}{\SingleTriple{\nodelabelvar}{\annquery{1}\conjunction\annquery{2}}{\nodelabelvar'}}}
      &\textsc{(TConj)}
      \\
      \inference
      {\judg{\queryTerm{1}}{\triplevar}}
      {\judg{\queryTerm{1}\union\queryTerm{2}}{\triplevar}}
      &\textsc{(TUnionL)}&
      \inference
      {\judg{\queryTerm{1}}{\SingleTriple{\nodelabelvar}{\annquery{1}}{\nodelabelvar'}} &
       \judg{\queryTerm{2}}{\SingleTriple{\nodelabelvar}{\annquery{2}}{\nodelabelvar''}}}
      {\judg{\branch{\queryTerm{1}}{\queryTerm{2}}}{\SingleTriple{\nodelabelvar}{\branch{\annquery{1}}{\annquery{2}}}{\nodelabelvar''}}}
      &\textsc{(TBranchL)}&
      \inference
      {\triplevar\in\PlusComp{}(\queryTerm{}, \typ{\queryTerm{}})}
      {\judg{\transclos{\queryTerm{}}}{\triplevar}}
      &\textsc{(TPlus)}
    \end{align*}
    \caption{Inference rules for the path expression-graph schema
      triple compatibility relation.}  
    \label{fig:TypeSystemLQ}
\end{figure*}

The set of triples compatible with a given path expression
$\queryTerm{}$ according to a schema $\Schema{}$ can be computed
inductively from the sets of triples compatibles with the parts of
$\queryTerm{}$. We do this using inference rules and an auxiliary function.

\begin{definition}[path expression-triple compatibility]
  Let $\queryTerm{}$ be a path expression, $\Schema{}$ a graph schema and $\triplevar$ a graph schema triple. The judgement
  $\judg{\queryTerm{}}{\triplevar}$ means that under schema $\Schema{}$,
  $\queryTerm{}$ is \emph{compatible} with $\triplevar$. It is defined
  inductively by the rules in Fig.~\ref{fig:TypeSystemLQ}, where
  $\typ{\queryTerm{}}$ represents the set of all triples compatible
  with $\queryTerm{}$: 
  \small $$\typ{\queryTerm{}} = \{\triplevar\mid\judg{\queryTerm{}}{\triplevar}\}$$
  \normalsize
\end{definition}

The last rule, \textsc{TPlus}, relies on the auxiliary function
$\PlusComp{}$, which works on the whole set of triples compatible with
$\queryTerm{}$ at once. This function, defined below, generates two
kinds of triples: (a) triples where the path expression is
$\transclos{\queryTerm{}}$ itself, with all annotations dropped, and
(b) triples where the path expression does not contain +, the
transitive closure operator, anymore. The rationale is that if the schema information allows us to avoid the costly\footnote{In terms of time and 
computing resources required for querying graph databases~\cite{hellings2018tarski,hellings2023expressive}.} 
transitive closure, we prefer to do so. However, we do not want to overcomplicate the expression if we cannot avoid the transitive closure operation.

To determine when transitive closure can be avoided, we associate to the set
$\typ{\queryTerm{}}$ the \emph{directed graph} whose vertices are the node
labels and whose edges are the triples. Any path of length $n$ in that
graph yields a triple compatible with $\queryTerm{}^n$ (as we can
infer by repeated application of \textsc{TConcat}). Since the meaning
of $\transclos{\queryTerm{}}$ is the union of the $\queryTerm{}^n$ for
all $n\geq 1$, the set of all (non-empty) paths in the graph
corresponds to triples compatible with $\transclos{\queryTerm{}}$. If
this set is finite, we actually define it as \emph{the} set of triples
compatible with $\transclos{\queryTerm{}}$. If it is infinite, then
transitive closure cannot be removed. Since the existence of
infinitely many paths is equivalent to the presence of cycles in the graph,
we propose the following algorithm for computing $\PlusComp{}$:

\begin{definition}[Plus Compatibility]\label{def:PlC}
  $\PlusComp{}(\queryTerm{}, \SchemaTriple{})$ is the set of triples
  resulting of the following:
  \begin{enumerate}
  \item Let $G$ be the directed graph associated with
    $\SchemaTriple{}$;
  \item Let $K$ be the set of vertices of $G$ which are part of a
    cycle;
  \item Let $R$ be the initially empty result set;
  \item For each path $p$ from $A$ to $B$ without cycles in $G$:
    \begin{itemize}
    \item if any of the vertices of $p$ (including $A$ and $B$) is in
      $K$, then add to $R$ the triple $\SingleTriple{A}{\transclos{\queryTerm{}}}{B}$
    \item else
      \begin{itemize}
      \item let $\psi$ be the annotated path expression
        resulting of concatenating all triples in $p$;
      \item add to $R$ the triple $\SingleTriple{A}{\psi}{B}$
      \end{itemize}
    \end{itemize}
  \item Return $R$.
  \end{enumerate}
\end{definition}

\begin{example}\label{examp:RuleApplication}
Consider the path expression $\queryTerm{4} = \texttt{livesIn} \concatenation
\texttt{isLocatedIn}^{+} \concatenation\texttt{dealsWith}^{+}$, and
let $\Schema{}$ be the schema of Fig.~\ref{fig:YAGOSchema}.
 Table~\ref{tab:RuleApplication} shows the sets of triples associated
to $\queryTerm{4}$  and its
three sub-terms $\queryTerm{1} = \texttt{livesIn} (\texttt{lvIn})$,
$\queryTerm{2} = \texttt{isLocatedIn}^{+} (\texttt{isL}^{+})$ and
$\queryTerm{3} = \texttt{dealsWith}^{+} (\texttt{dw}^{+})$.

For $\queryTerm{3}$, we have
$\typ{\texttt{dealsWith}} = \{(\texttt{COUNTRY}, \texttt{dealsWith},
\texttt{COUNTRY})\}$. The graph associated to that singleton set has
one vertex and one edge which forms a cycle, therefore the transitive
closure cannot be eliminated: we have
$\typ{\texttt{dealsWith}^{+}} = (\texttt{COUNTRY},
\texttt{dealsWith}^{+}, \texttt{COUNTRY})$.

For $\queryTerm{2}$, the set
$\typ{\texttt{isLocatedIn}}$ contains 3 triples; the associated graph
has 4 vertices (\texttt{PROPERTY}, \texttt{CITY}, \texttt{REGION} and
\texttt{COUNTRY}) and 3 edges corresponding to the 3 triples. It
contains no cycle, therefore
$\typ{\texttt{isLocatedIn}^{+}}$ contains 6 triples, corresponding to
the 6 non-empty paths in the graph.

Notice
that, while \textsc{TPlus} can yield many triples when it is possible to
remove the transitive closure, \textsc{TConcat} on the other hand can
drastically reduce their number, so that in the end there is only one
triple compatible with $\queryTerm{4}$.
\end{example}

\begin{table}[htbp]
  \centering
  \caption{Application of the inference rules of
    Fig.~\ref{fig:TypeSystemLQ} to the term $\queryTerm{4} = \texttt{livesIn} \concatenation \texttt{isLocatedIn}^{+} \concatenation\texttt{dealsWith}^{+}$.}  
  \label{tab:RuleApplication}
  \resizebox{.85\linewidth}{!}{%
  \begin{tabular}{c>{\raggedright}p{9cm}c}
    \toprule
    \textbf{TERM} & \textbf{TRIPLES}  & \textbf{RULE} \\
    \midrule
    \texttt{lvIn}   & \big(\texttt{PER}, \texttt{lvIn}, \texttt{CITY}\big) & \textsc{TBasic} \\
    \hline 
    $\texttt{isL}^{+}$ & \big(\texttt{PRO}, \texttt{isL}, \texttt{CITY}\big), \big(\texttt{CITY}, \texttt{isL}, \texttt{REG}\big),\big(\texttt{REG}, \texttt{isL}, \texttt{CUN}\big),\\ \big(\texttt{PRO}, \texttt{isL$\annconcat{\texttt{CITY}}$isL}, \texttt{REG}\big), \big(\texttt{PRO}, \texttt{isL$\annconcat{\texttt{CITY}}$isL$\annconcat{\texttt{REG}}$isL}, \texttt{CUN}\big), \\\big(\texttt{CITY}, \texttt{isL$\annconcat{\texttt{REG}}$isL}, \texttt{CUN}\big)& \textsc{TPlus}\\  
    \hline 
    $\texttt{dw}^{+}$      & \big(\texttt{CUN}, $\texttt{dw}^{+}$, \texttt{CUN}\big)  & \textsc{TPlus} \\
    \hline  
    $\texttt{lvIn} \concatenation \texttt{isL}^{+}$ & \big(\texttt{PER}, \texttt{lvIn$\annconcat{\texttt{CITY}}$isL}, \texttt{REG}\big), \big(\texttt{PER}, \texttt{lvIn$\annconcat{\texttt{CITY}}$isL$\annconcat{\texttt{REG}}$isL}, \texttt{CUN}\big)  & \textsc{TConcat}\\
    \hline   
    $\texttt{lvIn} \concatenation \texttt{isL}^{+} \concatenation \texttt{dw}^{+}$ & \big(\texttt{PER}, \texttt{lvIn$\annconcat{\texttt{CITY}}$isL$\annconcat{\texttt{REG}}$isL$\annconcat{\texttt{CUN}}$dw}$^{+}$, \texttt{CUN}\big)  & \textsc{TConcat}\\
    \bottomrule
      & \textbf{\texttt{isL=isLocatedIn, lvIn=livesIn,} \texttt{REG=REGION}} & \\
     & \textbf{\texttt{dw=dealsWith,} \texttt{PER=PERSON,} \texttt{PRO=PROPERTY,} \texttt{CUN=COUNTRY}} & \\
  \end{tabular}
  }
\end{table}

\subsubsection{Properties of the compatibility relation}
The path expression-triple compatibility relation is defined relative
to a schema. Consider now a database conforming to this schema. The
compatibility relation enjoys the following properties relative to any
such database:
\begin{itemize}
\item Soundness, meaning that whenever our path expression
  $\queryTerm{}$ is compatible with some triple
  ${\SingleTriple{\nodelabelvar}{\annquery{}}{\nodelabelvar'}}$, then
  all pairs of a node labeled $\nodelabelvar$ and a node labeled
  $\nodelabelvar'$ linked, in the database, by a path conforming to
  $\annquery{}$, are part of the result of $\queryTerm{}$;
\item Completeness, meaning that whenever $\queryTerm{}$ returns a
  pair of a node labeled $\nodelabelvar$ and a node labeled
  $\nodelabelvar'$, there exists a triple
  ${\SingleTriple{\nodelabelvar}{\annquery{}}{\nodelabelvar'}}$,
  compatible with $\queryTerm{}$, such that these nodes are linked, in
  the database, by a path conforming to $\annquery{}$.
\end{itemize}
These two properties make the initial path expression semantically
equivalent to its set of compatible triples, so long as we only
consider databases conforming to the schema. Formally:

\begin{theorem}\label{th:soundcomplete}
  Let $\Schema{}$ be a graph schema, let $\Database{}$ be a
  graph database conforming to $\Schema{}$ \emph{via} the
  schema-database mapping $\SchemaDatabaseMap{}$, and let
  $\queryTerm{}$ be a path expression.

\textsc{Soundness.}  Assume
  $\judg{\queryTerm{}}{\SingleTriple{\nodelabelvar}{\annquery{}}{\nodelabelvar'}}$
  holds for some triple
  $\SingleTriple{\nodelabelvar}{\annquery{}}{\nodelabelvar'}$. Let
  $(s,t)\in\sem{\annquery{}}$ such that $\NodeLabeling{\Database{}}(s) =
  \nodelabelvar$ and $\NodeLabeling{\Database{}}(t) = \nodelabelvar'$. Then $(s,t)\in\sem{\queryTerm{}}$.

  \textsc{Completeness.} Let $(s, t)\in\sem{\queryTerm{}}$. Then there
  exists $\annquery{}$ such that  $(s, t)\in\sem{\annquery{}}$ and
  $\judg{\queryTerm{}}{\SingleTriple{\NodeLabeling{\Database{}}(s)}{\annquery}{\NodeLabeling{\Database{}}(t)}}$.
\end{theorem}

\begin{proof} Both properties are proved by structural
  induction on $\queryTerm{}$. The base cases, where $\queryTerm{}$ is
  a single edge label, are a direct consequence of the consistency
  between $\Database{}$ and $\Schema{}$ (def.~\ref{def:SDConsistency})
  and of the definition of basic graph schema triples
  (def.~\ref{def:BGSTriples}). The inductive cases other than transitive
  closure are straightforward since the inference rules of
  Fig.~\ref{fig:TypeSystemLQ} follow exactly the structure of the
  semantics defined in Fig.~\ref{fig:intercqtGrammar}. Finally, we
  detail the case of transitive closure:

  \textbf{Soundness:} assume the property holds for $\queryTerm{}$.
    Let 
    $\SingleTriple{\nodelabelvar}{\annquery{}}{\nodelabelvar'} \in
    \PlusComp{}(\queryTerm{}, \typ{\queryTerm{}})$
    and let $(s,t)\in\sem{\annquery{}}$ such that $\NodeLabeling{\Database{}}(s) =
  \nodelabelvar$ and $\NodeLabeling{\Database{}}(t) = \nodelabelvar'$.
  We have two cases: if $\annquery{} = \transclos{\queryTerm{}}$ then
  the result is immediate. Otherwise, according to Def.~\ref{def:PlC},
  it means that there exists a sequence of triples $t_1,\ldots,t_n\in
  \typ{\queryTerm{}}$ such that: $\tripleSrc{}(t_1) = \nodelabelvar$,
  $\tripleTrg{}(t_i) = \tripleSrc{}(t_{i+1})$ for all $i < n$,
  $\tripleTrg{}(t_n) = \nodelabelvar'$, and $\annquery{}$ is the
  concatenation of $t_1\ldots t_n$ (i.\ e.\ 
  $\annquery{} = \tripleEdgeLabel{}(t_1)\annconcat{\tripleTrg{}(t_1)}\tripleEdgeLabel{}(t_2)\annconcat{\tripleTrg{}(t_2)}\cdots\annconcat{\tripleTrg{}(t_{n-1})}\tripleEdgeLabel{}(t_n)$). Since $(s, t)$ is in $\sem{\annquery{}}$, it means that
  there is a path in $\Database{}$ from $s$ to $t$ which conforms
  to $\annquery{}$. Since $\annquery{}$ is the concatenation of
  $t_1\ldots t_n$, that path has to be the concatenation of $n$ sub-paths
  conforming respectively to the triples $t_1\ldots t_n$. Since each
  of those triples is compatible with $\queryTerm{}$, we can use the
  induction hypothesis for each of them and see that the
  source-target pair for each one is in $\sem{\queryTerm{}}$. It
  results that $(s, t)$ is in $\sem{\queryTerm{}^n}$, and therefore in
  $\sem{\transclos{\queryTerm{}}}$.

  \textbf{Completeness:} assume the property holds for $\queryTerm{}$.
  Let $(s, t)\in\sem{\transclos{\queryTerm{}}}$. By definition, there
  exists $k\geq 1$ such that $(s, t)\in\sem{\queryTerm{}^k}$.
  This means that there is a sequence $n_0\ldots n_k$ of nodes in $D$
  with $n_0 = s$ and $n_k = t$ such that we have
  $(n_{i-1}, n_i)\in\sem{\queryTerm{}}$ for all $i$ between $1$ and
  $k$. By induction hypothesis, for each $i$ there is a triple
  $t_{i} =
  \SingleTriple{\NodeLabeling{\Database{}}(n_{i-1})}{\annquery{i}}{\NodeLabeling{\Database{}}(n_i)}$
  such that $\judg{\queryTerm{}}{t_i}$ and
  $(n_{i-1}, n_i)\in\sem{\annquery{i}}$. Let $G$ and $K$ be the graph
  and set of vertices defined in Def.~\ref{def:PlC}. The sequence
  $t_1\ldots t_k$ forms a path from $\NodeLabeling{\Database{}}(s)$ to
  $\NodeLabeling{\Database{}}(t)$ in $G$. If this path contains no
  cycle, then we see from step (4) of Def.~\ref{def:PlC} that
  $\PlusComp{}(\queryTerm{}, \typ{\queryTerm{}})$ contains a triple
  $\SingleTriple{\NodeLabeling{\Database{}}(s)}{\annquery{}}{\NodeLabeling{\Database{}}(t)}$
  with $\annquery{}$ equal either to $\transclos{\queryTerm{}}$ or to
  the concatenation of $\annquery{1}\ldots\annquery{k}$. In both
  cases, we have $(s, t)\in\sem{\annquery{}}$. If the path
  $t_1\ldots t_k$ does contain a cycle, say $\NodeLabeling{\Database{}}(n_i) = \NodeLabeling{\Database{}}(n_j)$ with $j>i$,
  then the sequence $t_1\ldots t_i,t_{j+1}\ldots t_k$ is also a path
  in $G$ without that cycle, and we have
  $\NodeLabeling{\Database{}}(n_j)\in K$ since
  $t_{i+1}\ldots t_j$ is a cycle in $G$. If that path still contains a
  cycle, we can reiterate this `shortcutting' until there are
  none left, and obtain a path from $\NodeLabeling{\Database{}}(s)$ to
  $\NodeLabeling{\Database{}}(t)$ without cycles and
  containing at least one vertex in $K$. Therefore,
  $\PlusComp{}(\queryTerm{}, \typ{\queryTerm{}})$ must contain the
  triple
  $\SingleTriple{\NodeLabeling{\Database{}}(s)}{\transclos{\queryTerm{}}}{\NodeLabeling{\Database{}}(t)}$,
  which allows us to conclude since we already have
  $(s,t)\in\sem{\transclos{\queryTerm{}}}$.
\end{proof}

\subsection{From triples to rewritten query}\label{subsec:mergeAnnPE}
\subsubsection{Merging triples}
\nobreak 

The compatibility relation allows us to obtain, from a path expression
$\queryTerm{}$
and a schema $\Schema{}$, a set $\typ{\queryTerm{}}$ of triples which
represent all the paths, annotated with node labels,
which can possibly contribute to the query result. In principle, we
could then convert each individual triple from $\typ{\queryTerm{}}$
into a (very specific) query of its own, combine all those queries
with a union and obtain a query which, provided the database conforms
to the schema, is equivalent to the initial path expression, while
avoiding the computation of intermediary results we know from the
schema are useless.
However, when several triples have the same underlying path expression
and differ only by the node labels, it would not be efficient to
compute their results separately and do the union afterwards. So, we
merge such triples by replacing all single \SIGMODchange{node labels} with sets of
possible node labels.

\begin{definition}
  Let $T$ be a set of graph schema triples with the same underlying
  path expression $\phi$. The \emph{merged triple} of $T$ is the triple
  $M(T) = (L_1, \Psi, L_2)$ where $L_1 = \{l\mid (l, \_, \_)
  \in T\}$, $L_2 = \{l\mid (\_, \_, l)\in T\}$, and $\Psi$ is the path
  expression $\phi$ where each concatenation step is annotated with
  the set of all labels annotating the same concatenation step in
  $\{\psi\mid (\_, \psi, \_)\in T\}$.
\end{definition}

\begin{example}\label{example:duplicateElimination1}
  Consider an initial path expression $\queryTerm{1} = a^{+} \concatenation b \concatenation d$. Suppose the set of triples $\typ{\queryTerm{}}$ 
  contains the following: $(m, a^{+}\annconcat{n} b\annconcat{l} d,
  p)$ and $(m, a^{+}\annconcat{q} b\annconcat{r} d, l)$. These triples
  can be merged into $(\{m\}, a^{+}\annconcat{\{n, q\}}
  b\annconcat{\{l, r\}} d, \{p, l\})$.
\end{example}

From the set $\typ{\queryTerm{}}$, we compute the set of merged
triples of $\phi$, $\mathcal{M}_{\Schema{}}(\queryTerm{})$, by: 1.
partitioning $\typ{\queryTerm{}}$ in subsets having the same
underlying path expression, and 2. taking the merged triple of each subset.
Note that in many cases, this will result in one single merged
triple, whose underlying path expression is $\phi$ itself, but it is
not the case whenever $\phi$ contains a union, or a transitive closure
which could be removed.

\subsubsection{Removing redundant annotations}

In some cases, the source and target node labels assigned to a
specific edge label in an annotated path expression are the only
source and target node labels associated with that same edge label in
the graph schema. In such a case, the annotation would not be useful
for minimizing intermediate result size since the whole dataset
already conforms to the annotation: on the contrary, it would
introduce a useless filtering step. Therefore, as a last step before
constructing the rewritten query, we detect and remove such annotations.

\begin{example}\label{example:duplicateElimination2}
  Let us assume that in a graph schema, source and target nodes 
  associated with edges labeled as $d$ and $b$ are $l$ and $r$ respectively. The node labels $l$ and $r$ can be eliminated from the annotated path expression in Ex.~\ref{example:duplicateElimination1}, 
  resulting in $(\{m\}, a^{+}\annconcat{\{n, q\}} b\concatenation d,
  \{p, l\})$. 
\end{example}

Note that, after merging triples and removing redundant annotations,
some queries may revert to the initial path expression. This means
that the schema did not contain information which would allow
optimizing the query.
  
\subsubsection{Rewritten queries}

\begin{figure}
  \centering
  \footnotesize
  \begin{align*}
    &\queryOf(\alpha, \beta, \queryTerm{}) =
    (\emptyset, \emptyset, \{(\alpha, \queryTerm{}, \beta)\}) \\
    &\queryOf(\alpha, \beta, \annquery{1}\annconcat{L}\annquery{2}) =
    \text{let $\gamma$ be a fresh variable,}\\&\quad\text{let } (\Body{1},\Atomic{1}, \CQTRelation{1}) = \queryOf(\alpha, \gamma, \annquery{1})\text{ and } (\Body{2},\Atomic{2}, \CQTRelation{2}) = \queryOf(\gamma, \beta, \annquery{2})\\&\quad\text{in }
    (\Body{1}\cup\Body{2}\cup\{\gamma\}, \Atomic{1}\cup\Atomic{2}\cup\{\NodeLabeling{\Atomic{}}(\gamma) \in L\}, \CQTRelation{1}\cup\CQTRelation{2})\\
    &\queryOf(\alpha, \beta, \annquery{1}\branch{\annquery{2}}) =
    \text{let $\gamma$ be a fresh variable,}\\&\quad\text{let } (\Body{1},\Atomic{1}, \CQTRelation{1}) = \queryOf(\alpha, \beta, \annquery{1})\text{ and } (\Body{2},\Atomic{2}, \CQTRelation{2}) = \queryOf(\beta, \gamma, \annquery{2})\\
    &\quad\text{in } (\Body{1}\cup\Body{2}\cup\{\gamma\}, \Atomic{1}\cup\Atomic{2}, \CQTRelation{1}\cup\CQTRelation{2})\\
    &\queryOf(\alpha, \beta, \branch{\annquery{1}}{\annquery{2}}) =
    \text{let $\gamma$ be a fresh variable,}\\&\quad\text{let } (\Body{1},\Atomic{1}, \CQTRelation{1}) = \queryOf(\alpha, \gamma, \annquery{1})\text{ and } (\Body{2},\Atomic{2}, \CQTRelation{2}) = \queryOf(\alpha, \beta, \annquery{2})\\
    &\quad\text{in } (\Body{1}\cup\Body{2}\cup\{\gamma\}, \Atomic{1}\cup\Atomic{2}, \CQTRelation{1}\cup\CQTRelation{2})\\
    &\queryOf(\alpha, \beta, {\annquery{1}}\conjunction{\annquery{2}}) =
    \\&\quad\text{let } (\Body{1},\Atomic{1}, \CQTRelation{1}) = \queryOf(\alpha, \beta, \annquery{1})\text{ and } (\Body{2},\Atomic{2}, \CQTRelation{2}) = \queryOf(\alpha, \beta, \annquery{2})\\
    &\quad\text{in } (\Body{1}\cup\Body{2}, \Atomic{1}\cup\Atomic{2}, \CQTRelation{1}\cup\CQTRelation{2})
  \end{align*}
  
  \caption{Translation of annotated path expressions into \cqt{}s.}
  \label{fig:translationToCQT}
\end{figure}

We now want to translate merged triples into $\cqt$ queries. 
For this, we first observe that the annotated path expressions
$\annquery{}$ generated by the rules of Fig.~\ref{fig:TypeSystemLQ}
and Def.~\ref{def:PlC} always obey some syntactic restrictions,
namely: no annotations appear under the transitive closure operator, and the union operator never appears, except possibly under
  transitive closure.
Furthermore, the
reverse operation only occurs in front of edge labels. These
observations allow us to reduce the number of cases in the
translation: $\annquery{}$ is either a plain path expression, a concatenation, a branching or a conjunction.

\begin{definition}[\cqt{} of a merged triple]
The translation is defined using a function $\queryOf(\alpha, \beta,
\annquery{})$, where $\alpha$ and $\beta$ are \cqt{} variables and
$\annquery{}$ an annotated path expression, which returns a triple
$(\Body{}, \Atomic{}, \CQTRelation{})$ corresponding to a \cqt{}
representing $\annquery{}$ with $\Head{} = \{\alpha,
\beta\}$ (see Def.~\ref{def:cqt}). This function is defined in
Fig.~\ref{fig:translationToCQT}.

For a given merged triple $\triplevar = \SingleTriple{L}{\annquery{}}{L'}$, we then define the associated \cqt{} as $\CQTQuery{}(\triplevar) = (\{\alpha, \beta\}, \Body{}, \Atomic{}\cup\{\NodeLabeling{\Atomic{}}(\alpha) \in L, \NodeLabeling{\Atomic{}}(\beta) \in L'\}, \CQTRelation{})$ where $(\Body{}, \Atomic{}, \CQTRelation{}) = \queryOf(\alpha, \beta, \annquery{})$.
\end{definition}

This allows us to associate to any path expression, given a schema
$\Schema{}$, a \ucqt representing the query enriched with schema
information:

\begin{definition}[schema-enriched query]
  Given a path expression $\queryTerm{}$ and a schema $\Schema{}$, the
  schema-enriched query of $\queryTerm{}$,
  $\mathcal{R}_{\Schema{}}(\queryTerm{})$, is the following \ucqt: 
  $\mathcal{R}_{\Schema{}}(\queryTerm{}) =  \{\alpha, \beta\}\leftarrow \{\bigcup_{\triplevar\in\mathcal{M}_{\Schema{}}{\queryTerm{}}}\CQTQuery{}(\triplevar)\}$
\end{definition}

Thanks to Theorem~\ref{th:soundcomplete}, we have that $\queryTerm{}$
is equivalent to $\mathcal{R}_{\Schema{}}(\queryTerm{})$ on any
database conforming strictly to $\Schema{}$.

\begin{example} Consider the path expression $\queryTerm{4} =
  \texttt{lvIn} \concatenation \texttt{isL}^{+} \concatenation
  \texttt{dw}^{+}$ of Example~\ref{examp:RuleApplication}. We saw that
  our inference system derives exactly one triple compatible with it
  in schema $\Schema{}$: \big(\texttt{PER},
  \texttt{lvIn$\annconcat{\texttt{CITY}}$isL$\annconcat{\texttt{REG}}$isL$\annconcat{\texttt{CUN}}$dw}$^{+}$,
  \texttt{CUN}\big). Then, since there is only one triple, there is
  nothing to merge. We then remove redundant labels: the schema
  $\Schema{}$ of Fig.~\ref{fig:YAGOSchema} only allows
  \textsf{livesIn} from \textsf{PERSON} to \textsf{CITY} and
  \textsf{dealsWith} from \textsf{COUNTRY} to \textsf{COUNTRY}, so the
  final unique merged triple is: \big($\emptyset$,
  \texttt{lvIn$\concatenation$isL$\annconcat{\{\texttt{REG}\}}$isL$\concatenation$dw}$^{+}$, $\emptyset$\big).
  Therefore, we can rewrite this expression into:
  \small 
  \begin{align*}
    \mathcal{R}_{\Schema{}}(\queryTerm{4}) = \{\alpha, \beta \mid
  \exists \gamma\quad &(\alpha, \texttt{lvIn}\concatenation\texttt{isL}, \gamma)\wedge(\gamma, \texttt{isL}\concatenation\texttt{dw}^{+}, \beta) \wedge \NodeLabeling{\Atomic{}}(\gamma) \in\{ \texttt{REG}\}
  \}\end{align*}
  \normalsize
We can notice that the rewritten query allows the database engine to
avoid computing the transitive closure of \texttt{isL} at all.
\end{example}

\section{System Implementation}\label{sec:Implementation}
\nobreak

We analyse the performance of queries augmented with schema information. The approach has been implemented as a system whose architecture is depicted in Figure~\ref{fig:SysArch}. Our system architecture consists of three main modules~\emph{(i) Rewriter},~\emph{(ii) Translator} and~\emph{(iii) Backend}. 

\begin{figure}[htbp]
    \centering
    \includegraphics[width=.7\linewidth]{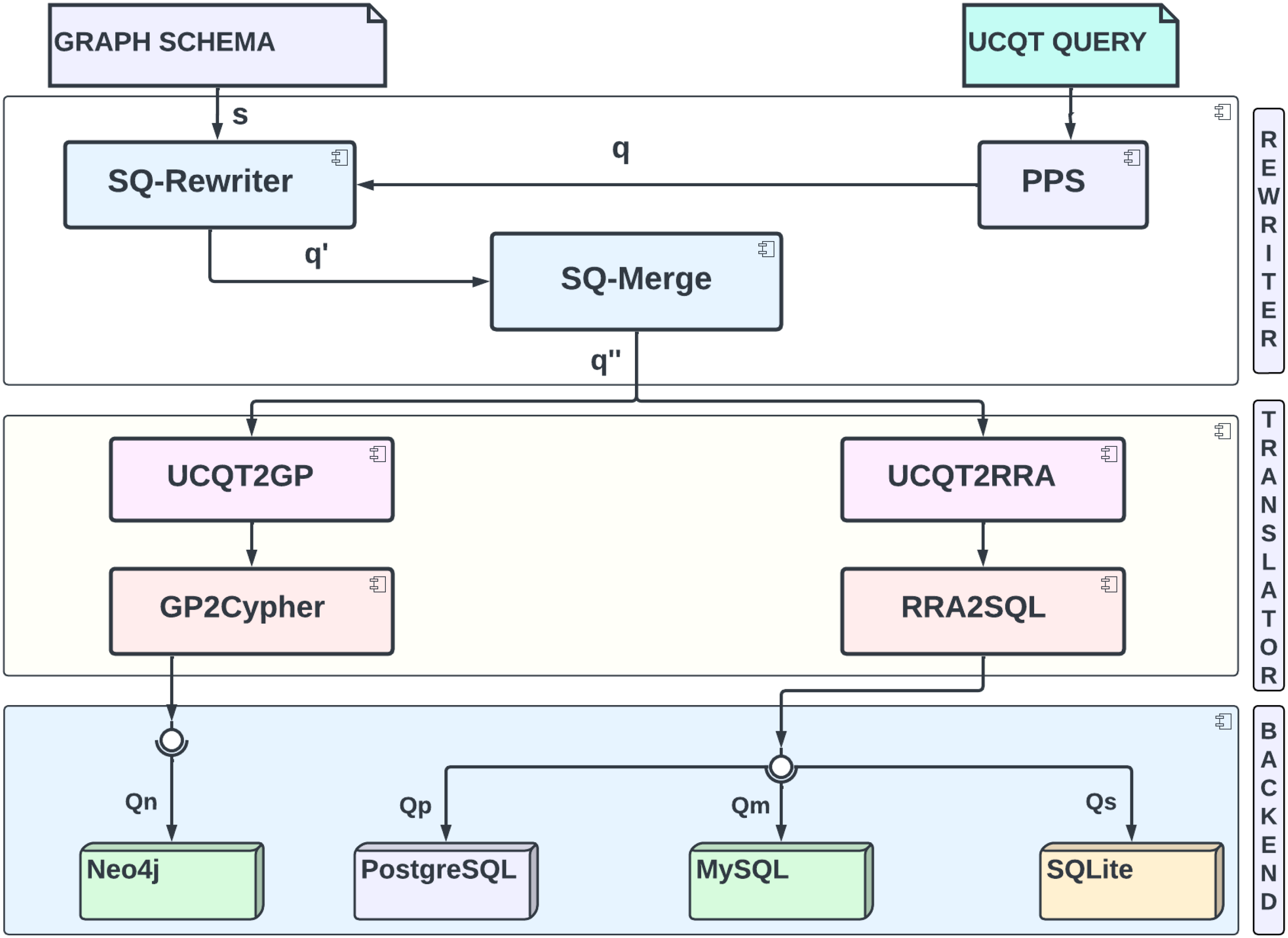}
    \caption{System architecture.}
    \label{fig:SysArch}
\end{figure}

\paragraph{\underline{Rewriter}}
The rewriter module consists of three components~\emph{Preliminary Path Simplifier} (PPS),~\emph{Schema-based Query Rewriter} (SQ-Rewriter) and~\emph{Schema-based Query Merge} (SQ-Merge). 
The PPS component takes a $\ucqt{}$ query as input, then applies the preliminary path simplification rules presented in Fig.~\ref{fig:simpleRWrules} and produces a simplified $\ucqt{}$ query $q$ as output. 
The SQ-Rewriter component implements the query rewritings leveraging schema information presented in Sec.~\ref{sec:logicalQueryOpt} for the $\ucqt$ query language (defined in Section~\ref{sec:queryingPG}). 
It takes a $\ucqt$ query $q$ and a graph schema $s$ described in the graph schema formalism (using Def.~\ref{def:PGSchema}) as input and generates a schema-aware $\ucqt$ query $q'$ as output.
The SQ-Merge component uses the graph schema to remove redundant annotations by merging path expressions as described in Sec.~\ref{subsec:mergeAnnPE}. It takes a schema-aware $\ucqt$ query $q'$ as input and produces 
merged $\ucqt$ query $q''$ as output.

\paragraph{\underline{Translator}}
 
The translator module compiles an $\ucqt$ query into a graph query that can be executed by a graph database management system (GDBMS) and a recursive SQL query that can be executed by a relational database management system (RDBMS).
The translator module uses the UCQT2GP component to convert the $\ucqt{}$ query into a union of graph patterns, which is then transformed into a Cypher query by the GP2Cypher component. 
Due to the limited expressive power offered by Cypher~\cite{sharma2021practical}, only \uctworpq{} graph patterns (not $\ucqt$) can be converted into Cypher. 
\SIGMODchange{Since we use the \ucqt{} graph query language formalism to express graph patterns, our approach extends beyond Cypher: it can also be applied to other practical graph query languages 
such as SPARQL and PGQL~\cite{sharma2021practical}. If one starts from a query written in any of these formalisms, it can easily be translated into $\ucqt$ before being given as input to our system.}

The translator module uses the UCQT2RRA component to translate $\ucqt{}$ query into a term in recursive relational algebra (RRA), and that is further optimised using the $\mRA$ system based on the 
rewritings proposed in~\cite{jachiet2020optimization}. Compared to the path expressions considered in~\cite{jachiet2020optimization}, we consider additional rules for translating~\emph{conjunction} and~\emph{branching} 
operations as shown in Tab.~\ref{tab:muratranslation}. 

\begin{table}[htbp]
  \caption{Conjunction and branching translation to Recursive Relational Algebra.}
  \label{tab:muratranslation}
  \centering
  \resizebox{.8\linewidth}{!}{%
  \begin{tabular}{cc}
    \toprule
     \textbf{Path expression} & \textbf{Recursive Relational Algebra term}\\
    \midrule
     $\llparenthesis\queryTerm{1} \conjunction \queryTerm{2}\rrparenthesis =$ & $\big \{ \rename n \trgCol {\rename m \srcCol {\rename \trgCol n {\rename \srcCol m {\muTerm{1}}} \NJoin \rename \trgCol n {\rename \srcCol m {\muTerm{2}}}}} \big | \quad \muTerm{1} \in \llparenthesis\queryTerm{1}\rrparenthesis \wedge \muTerm{2} \in \llparenthesis\queryTerm{2}\rrparenthesis \big \}$ \\
     $\llparenthesis\queryTerm{1}[\queryTerm{2}] \rrparenthesis   =$& $\big \{ \rename m \trgCol {\rename \trgCol m {\muTerm{1}} \NJoin \rename \srcCol m {\proj \srcCol {\muTerm{2}}}} \big | \quad \muTerm{1} \in \llparenthesis\queryTerm{1}\rrparenthesis \wedge \muTerm{2} \in \llparenthesis\queryTerm{2}\rrparenthesis \big \}$ \\
     $\llparenthesis[\queryTerm{1}]\queryTerm{2}\rrparenthesis    =$ & $\big \{ \rename m \srcCol {\rename \srcCol m {\proj \srcCol {\muTerm{1}}} \NJoin \rename \srcCol m {\muTerm{2}}} \big | \quad \muTerm{1} \in \llparenthesis\queryTerm{1}\rrparenthesis \wedge \muTerm{2} \in \llparenthesis\queryTerm{2}\rrparenthesis \big \}$\\
  \bottomrule
\end{tabular}
}
\end{table}

Our system architecture is modular, allowing for components such as UCQT2RRA (currently based on the $\mRA$ system) to be replaced with other systems such as $\alpha$-extended RA~\cite{agrawal1988alpha}, $\beta$-RA~\cite{gomes2015beta}, \texttt{WAVEGUIDE}~\cite{yakovets2015waveguide}, \texttt{AVANTGRAPH}~\cite{leeuwen2022avantgraph} and \SIGMODchange{Datalog} based RRA systems~\cite{aref2015design,leone2006dlv,urbani2016vlog}. 
However, as discussed in~\cite{jachiet2020optimization,fejza2023mu}, the $\mRA$ system is superior in query optimisation capabilities.

The optimised RRA term is transformed into a concrete recursive SQL syntax for each RDBMS using the RRA2SQL component where the fixpoints are translated into~\emph{recursive view} 
statements\footnote{MySQL:~\texttt{CREATE OR REPLACE VIEW WITH RECURSIVE, SQLite:~\texttt{CREATE VIEW WITH RECURSIVE} \\ PostgreSQL:~\texttt{CREATE TEMPORARY RECURSIVE VIEW}}}.

\paragraph{\underline{Backend}} 
To evaluate our approach, we conducted experiments on one GDBMS and three RDBMS. Specifically, we used Neo4j, PostgreSQL, MySQL, and SQLite. Graph databases can be directly stored as property graphs on Neo4j. 
However, in the relational data model, we represent the nodes and edges of the graph database as relational tables.

\SIGMODchange{For instance, in the YAGO dataset's graph representation (as shown in Fig.~\ref{fig:YAGODatabase}), there are four graph edges labeled \texttt{isLocatedIn} connecting different pairs 
of nodes labeled as (\texttt{PROPERTY}, \texttt{CITY}), (\texttt{CITY}, \texttt{REGION}), and (\texttt{REGION}, \texttt{COUNTRY}) respectively. In the relational representation of the YAGO dataset, 
a table is created for each type of edge label, each with at least two columns $\srcCol$ and $\trgCol$, which are foreign keys pointing to the source and target nodes. Similarly, 
a table is created for each type of node label, with a specific column $\srcCol$ serving as a primary key and potentially many other columns for properties. Fig.~\ref{fig:FollowsRPG} 
shows two node tables (\emph{PROPERTY, CITY}) and two edge tables (\emph{owns, isLocatedIn}). Overall, each row in node or edge tables represents a node or an edge of the graph database.}

\begin{figure}[htbp]
    \centering
    \footnotesize 
     \begin{tikzpicture}
        [table/.style={matrix of nodes, nodes in empty cells, column sep=-0.35em, row sep=0}]
        \matrix(owns)[table, label=above: \textbf{owns},fill=yellow!5]{
        \hline     
        $\srcCol$ & $\trgCol$  \\
        \hline  
        $n_2$ & $n_1$\\    
        }; 
        \matrix(property) [table, label=above: \textbf{PROPERTY},fill=red!5,right =of owns] {
                 \hline
                 $\srcCol$ &  \texttt{address} \\
                 \hline  
                 $n_1$ &  7 Queen  \\
                       &   Street \\
                 };       
         \matrix(islocatedin) [table, label=above: \textbf{isLocatedIn}, fill=yellow!5,right =of property] {
             \hline 
             $\srcCol$ & $\trgCol$  \\
             \hline  
             $n_1$ & $n_6$\\
             $n_6$ & $n_5$\\
             $n_4$ & $n_5$\\
             $n_5$ & $n_7$\\
             };
 
             \matrix(area) [table, label=above: \textbf{CITY},fill=blue!5, right=of islocatedin] {
                 \hline
                 $\srcCol$ &  \texttt{name} \\
                 \hline  
                 $n_4$ & Elerslie \\
                 $n_6$ & Montbonnot\\ 
                 };     
     \end{tikzpicture}
     
     \caption{Relational representation of nodes and edges.}
     \label{fig:FollowsRPG}
 \end{figure}

\section{Experiments}\label{sec:experiments}
\nobreak 

We assess the impact on performance of the schema-based approach experimentally with a complete prototype implementation. Concretely, we try to answer the following questions:
(i) Does the schema-based approach improve the performance of query evaluation on real and synthetic benchmark datasets?
(ii) How does the approach behave on different database management systems?

\subsection{Experimental Setup}
\nobreak

\subsubsection{Datasets}
\nobreak 
We consider datasets of different nature: 
\begin{itemize}

  \item YAGO \cite{yago2019high}, which is a real knowledge graph. We use a preprocessed and cleaned version of the real-world dataset YAGO2s~\cite{yago2019high} in which only nodes with unique identifiers are present. 
  We split the set of RDF triples into multiple edge relations (tables), one for each predicate name. We create a node relation (table) for each node class. The dataset conforms to the YAGO schema constructed for this study. In other terms, the constructed schema does not exclude any existing triples from this YAGO dataset.  

  \item the Social Network Benchmark (SNB) interactive workload from the Linked Data Benchmark Council (LDBC) \cite{erling2015ldbc} which is a synthetic reference benchmark for property graphs. Specifically, we used the 
  LDBC-SNB dataset in CSV format provided from~\cite{cwi:snb}. 

\end{itemize}

\begin{table}[h]
  \caption{Summary of dataset characteristics.}
  \centering
  \label{tab:dataset}
  \resizebox{.6\linewidth}{!}{%
  \begin{tabular}{ccccrrr}
    \toprule
    \textbf{Name} & \textbf{SF} & \textbf{\#{}NR} & \textbf{\#{}ER} & \textbf{\#{}Nodes} & \textbf{\#{}Edges} & \textbf{Size}\\
    \midrule
    \textbf{YAGO}       &     N/A  &7                      &   88                 &    98,582 & 150,391,592  & 26 GB \\                         \hline 
                        &     0.1  &                       &                      &   416,311 &   2,034,983  & 0.3 GB\\
                        &     0.3  &                       &                      &  1,154,108 &   6,235,570  & 0.9 GB \\
                        &     1    &                       &                      &  3,966,203 &  23,056,025  & 3.3 GB \\
    \textbf{LDBC-SNB}   &     3    &       8               &   16                 & 11,407,480 &  69,482,982  & 9.9 GB \\
                        &     10   &                       &                      & 36,485,994 & 231,532,873  & 33 GB \\
                        &     30   &                       &                      & 88,789,972 & 541,279,759  & 82 GB \\
  \bottomrule
  
\end{tabular}
}
\end{table}

Tab.~\ref{tab:dataset} summarises the characteristics of the pre-processed datasets. %
The YAGO dataset has 98k nodes and 150M edges, with 7 node relations and 88 edge relations. 
LDBC-SNB uses 8 node relations (\textbf{NR}) and 16 edge relations (\textbf{ER}). 
LDBC-SNB has property graphs of varying sizes measured by scale factors (\textbf{SF}), ranging from 0.1 to 30. We consider 6 of them shown in Tab.~\ref{tab:dataset}. LDBC property graph with SF 0.1 has 416k nodes and 2M edges, 
SF 30 has 88M nodes and 541M edges~\footnote{LDBC-SNB contains large property graphs with scale factors 100, 1000}. The size column in Tab.~\ref{tab:dataset} shows the size on disk of the PostgreSQL database.

\subsubsection{Schemas}\label{subsubsec:preprocessing}
\nobreak

YAGO does not have a graph schema, therefore we developed a basic one, as illustrated in Fig.~\ref{fig:YAGOSchema}~\footnote{Due to space constraints, we only provide a small schema for YAGO in Fig.~\ref{fig:YAGOSchema}. 
The complete schema is available as supplementary material (see Sec.~\ref{subsubsec:Availability}).}, inspired by the SHACL semantic constraints from \cite{suchanek2023integrating} that specify the disjointness of certain classes.  %
The LDBC-SNB dataset comes with a pre-defined property graph schema~\cite{erling2015ldbc}.

\begin{table}[htbp]
  \caption{Queries for the LDBC-SNB Dataset.}
  
  \label{tab:LDBCQueries}
  \resizebox{.93\linewidth}{!}{%
  \begin{tabular}{clc}
    \toprule
    \textbf{Query Label} & \textbf{Path expressions as UCQT queries}  &   \textbf{Query Type} \\
    \midrule
    IC1 & \texttt{x1, x2 $\leftarrow$ (x1, \textcolor{violet}{\textbf{knows1..3\concatenation{}(isL $\cup$ (workAt $\cup$ studyAt)\concatenation{}isL)}}, x2)}                                                                                   & NQ \\
    IC2 & \texttt{x1, x2 $\leftarrow$ (x1, \textcolor{violet}{\textbf{knows\concatenation{}-hasC}}, x2)}                                                                                                                             & NQ \\
    IC6 & \texttt{x1, x2 $\leftarrow$ (x1, \textcolor{violet}{\textbf{knows1..2\concatenation{}(-hasC[hasT])[hasT]}}, x2)}                                                                                                          & NQ \\ 
    IC7 & \texttt{x1, x2 $\leftarrow$ (x1, \textcolor{violet}{\textbf{(-hasC\concatenation{}-likes) $\cup$ ((-hasC\concatenation{}-likes) $\conjunction{}$ knows)}}, x2)}                                                                 & NQ \\
    IC8 & \texttt{x1, x2 $\leftarrow$ (x1, \textcolor{violet}{\textbf{-hasC/-replyOf/hasC}}, x2)}                                                                                                                                & NQ \\
    IC9 & \texttt{x1, x2 $\leftarrow$ (x1, \textcolor{violet}{\textbf{knows1..2\concatenation{}-hasC}}, x2)}                                                                                                                        & NQ \\
    IC11 & \texttt{x1, x2 $\leftarrow$ (x1, \textcolor{violet}{\textbf{knows1..2\concatenation{}workAt\concatenation{}isL}}, x2)}                                                                                                   & NQ \\
    IC12 & \texttt{x1, x2 $\leftarrow$ (x1, \textcolor{violet}{\textbf{knows\concatenation{}-hasC\concatenation{}replyOf\concatenation{}hasT\concatenation{}hasTY\concatenation{}isSubC+}}, x2)}                                      & RQ \\
    IC13 & \texttt{x1, x2 $\leftarrow$ (x1, \textcolor{violet}{\textbf{knows+}}, x2)}                                                                                                                                                 & RQ \\
    IC14 & \texttt{x1, x2 $\leftarrow$ (x1, \textcolor{violet}{\textbf{(knows $\conjunction{}$ (-hasC/replyOf/hasC))+}}, x2)}                                                                                                         & RQ \\
    Y1    & \texttt{x1, x2 $\leftarrow$ (x1, \textcolor{violet}{\textbf{knows+\concatenation{}studyAt\concatenation{}isL+/isP+}}, x2)}                                                                                                  & RQ\\
    Y2   & \texttt{x1, x2 $\leftarrow$ (x1, \textcolor{violet}{\textbf{likes\concatenation{}hasC\concatenation{}knows+/isL+}}, x2)}                                                                                                   & RQ\\
    Y3   & \texttt{x1, x2 $\leftarrow$ (x1, \textcolor{violet}{\textbf{likes\concatenation{}replyOf+\concatenation{}isL+\concatenation{}isP+}}, x2)}                                                                                  & RQ\\
    Y4 & \texttt{x1, x2 $\leftarrow$ (x1, \textcolor{violet}{\textbf{hasM\concatenation{}(studyAt $\cup$ workAt)\concatenation{}isL+\concatenation{}isP+}}, x2)}                                                                             & RQ\\
    Y5 & \texttt{x1, x2 $\leftarrow$ (x1, \textcolor{violet}{\textbf{-hasM\concatenation{}([cof]hasT)\concatenation{}hasTY\concatenation{}isSubC+}}, x2)}                                                                             & RQ\\
    Y6 & \texttt{x1, x2 $\leftarrow$ (x1, \textcolor{violet}{\textbf{replyOf+\concatenation{}isL+\concatenation{}isP+}}, x2)}                                                                                                         & RQ\\
    Y7 & \texttt{x1, x2 $\leftarrow$ (x1, \textcolor{violet}{\textbf{hasMod\concatenation{}hasI\concatenation{}hasTY\concatenation{}isSubC+}}, x2)}                                                                                   & RQ\\
    Y8 & \texttt{x1, x2 $\leftarrow$ (x1, \textcolor{violet}{\textbf{([cof\concatenation{}hasC]hasM)\concatenation{}isL/isP+}}, x2)}                                                                                                  & RQ\\
    IS2 & \texttt{x1, x2 $\leftarrow$ (x1, \textcolor{violet}{\textbf{-hasC\concatenation{}replyOf+\concatenation{}hasC}}, x2)}                                                                                                   & RQ\\
    IS6 & \texttt{x1, x2 $\leftarrow$ (x1, \textcolor{violet}{\textbf{replyOf+\concatenation{}-cof\concatenation{}hasM}}, x2)}                                                                                                        & RQ\\
    IS7 & \texttt{x1, x2 $\leftarrow$ (x1, \textcolor{violet}{\textbf{(-hasC\concatenation{}replyOf\concatenation{}hasC) $\cup$ ((-hasC\concatenation{}replyOf\concatenation{}hasC) $\conjunction{}$ knows)}}, x2)}                           & NQ\\
    BI11 &\texttt{x1, x2 $\leftarrow$ (x1, \textcolor{violet}{\textbf{(([isL\concatenation{}isP]knows)[isL\concatenation{}isP]) $\conjunction{}$ (knows\concatenation{}([isL\concatenation{}isP]knows)))}} ,x2) }                  & NQ\\
    BI10 & \texttt{x1, x2 $\leftarrow$ (x1, \textcolor{violet}{\textbf{(knows+[isL\concatenation{}isP])\concatenation{}(-hasC[hasT])\concatenation{}hasT\concatenation{}hasTY}}, x2)}                                                 & RQ\\
    BI3 & \texttt{x1, x2 $\leftarrow$ (x1, \textcolor{violet}{\textbf{-isP\concatenation{}-isL\concatenation{}-hasMod\concatenation{}cof\concatenation{}-replyOf+\concatenation{}hasT\concatenation{}hasTY}}, x2)}                    & RQ\\
    BI9 & \texttt{x1, x2 $\leftarrow$ (x1, \textcolor{violet}{\textbf{replyOf+\concatenation{}hasC}}, x2)}                                                                                                                            & RQ\\
    BI20 & \texttt{x1, x2 $\leftarrow$ (x1, \textcolor{violet}{\textbf{(knows $\conjunction{}$ (studyAt\concatenation{}-studyAt))+}}, x2)}                                                                                             & RQ\\
    LSQB1 & \texttt{x1, x2 $\leftarrow$ (x1, \textcolor{violet}{\textbf{-isP\concatenation{}-isL\concatenation{}-hasM\concatenation{}cof\concatenation{}-replyOf+\concatenation{}hasT\concatenation{}hasTY}}, x2)}                    & RQ\\
    LSQB4 & \texttt{x1, x2 $\leftarrow$ (x1, \textcolor{violet}{\textbf{((likes[hasT])[-replyOf])\concatenation{}hasC}}, x2)}                                                                                                         & NQ\\
    LSQB5 & \texttt{x1, x2 $\leftarrow$ (x1, \textcolor{violet}{\textbf{-hasT\concatenation{}-replyOf\concatenation{}hasT}}, x2)}                                                                                                 & NQ\\
    LSQB6 & \texttt{x1, x2 $\leftarrow$ (x1, \textcolor{violet}{\textbf{knows\concatenation{}knows\concatenation{}hasI}}, x2)}                                                                                                     & NQ\\
    \bottomrule
                 
                        & \textbf{\texttt{isL=isLocatedIn,} \texttt{hasT=hasTag,} \texttt{isP=isPartOf,} \texttt{isSubC=isSubClassOf,} \texttt{hasI=hasInterest,}}                   & \\ 
                        & \textbf{\texttt{hasTY=hasType,} \texttt{coF=containerOf,} \texttt{hasMod=hasModerator,} \texttt{hasC=hasCreator,} \texttt{hasM=hasMember}}    & \\
  \end{tabular}
  }
\end{table}

\subsubsection{Queries}

We selected 30 queries from the LDBC-SNB workload~\cite{erling2015ldbc}, divided into two categories: non-recursive graph queries (\textbf{NQ}) and recursive graph queries (\textbf{RQ}). As shown in Tab.~\ref{tab:LDBCQueries} out of the 
30 queries, 12 are non-recursive, and 18 are recursive. 
Out of the 30 queries of Tab.~\ref{tab:LDBCQueries}, 22 are third-party queries. Queries labeled as (\textbf{IC} and \textbf{IS}) are extracted from the LDBC interactive workload, while queries labeled as (\textbf{BI}) are from the business intelligence workload, 
and queries labeled as (\textbf{LSQB}) are extracted from the large-scale subgraph query benchmark~\cite{mhedhbi2021lsqb}. Finally, we proposed queries labeled as (\textbf{Y}) as complementary queries, 
inspired by YAGO-style queries found in~\cite{jachiet2020optimization}.

We consider 18 queries on the YAGO dataset. These queries were previously used in studies such as~\cite{abul2017tasweet,gubichev2013sparqling,yakovets2015waveguide}. All the queries selected for the YAGO dataset are recursive graph queries. 

\subsubsection{Baseline}
\nobreak 

All experiments conducted using the schema-based approach are systematically compared to the initial (non-schema enriched) query considered as a baseline.

\subsubsection{Measures and timeout} 
\nobreak 
We set a 30-minute computation time limit for each query to evaluate schema-based and baseline approaches. If the computation exceeded 30 minutes, we stopped the process and deemed the query infeasible with the given approach, 
allowing us to measure the system's performance in a reasonable time frame. Each reported query execution time corresponds to the average of 5 runs. 

In experiments, some queries result in timeouts at a given scale factor but not with smaller dataset sizes. The primary focus is to measure the relative performance improvements brought by the proposed optimisations rather than conducting pure performance comparisons between systems. In particular, we want to check if such optimisations are capable of turning infeasible queries to feasible ones at some scale factors on the various systems.

\subsubsection{Hardware and software setup} 
\nobreak 
All experiments have been conducted by using a Macbook Air laptop with Apple M2 chip, 8 cores (4 performance and 4 efficiency), 24 GB of RAM and 1 TB hard disk. 
The main backend used for evaluation is PostgreSQL 15.4. We also provide performance comparisons with the graph database system Neo4j 5.21.2 community edition. 
We do not report on comparisons with MySQL and SQLite because both systems are significantly less efficient than PostgreSQL in executing queries for large LDBC datasets. No queries could be succesfully executed on MySQL and SQLite beyond scale factor 1.

\subsubsection{Availability}\label{subsubsec:Availability} 
The datasets, schemas, initial and rewritten queries are available\footnote{~\url{https://gitlab.inria.fr/tyrex-public/schema-graph-query.git}}.
 
\subsection{Query feasibility}

For the YAGO dataset, 17 queries out of 18 did run successfully with both approaches within the allowed timeout period. Only one query (query 7) timed out. Furthermore, during the process of schema enrichment and query factorization, query 7 automatically reverted to the 
initial query form since no schema-based optimization was found.  
The LDBC dataset consists of ten queries (\emph{IC2, IC6, IC7, IC9, IC13, Y7, BI11, BI9, BI20}, and~\emph{LSQB6}) that returned to their initial path expressions after going through the schema 
enrichment and query factorization process. Among these, six are non-recursive, while the remaining four are recursive. 

These queries highlight an advantage of our approach, which is to automatically avoid schema enrichment when it does not detect any significant performance gain. This allows us to prevent query performance degradation. 
For example, query 7 for YAGO and ten LDBC queries did not benefit from schema enrichment, but their performance was not made worse.

Out of ten LDBC queries, only three queries~\emph{Y7, BI11} and~\emph{BI20} could be executed on all six scale factors. 
Tab.~\ref{tab:feasibleQueries} summarizes the number and percentage of successful runs for the various scale factors (data size) on the LDBC dataset.

\begin{table}[h]
  \centering
  \caption{LDBC query feasibility across six scale factors (SF). }
  
  \label{tab:feasibleQueries}
  \resizebox{.7\linewidth}{!}{%
  \begin{tabular}{@{}ccccccccc@{}}
  \toprule
                      & \multicolumn{4}{c}{\textbf{Recursive}}                             & \multicolumn{4}{c}{\textbf{Non Recursive}}                         \\ 
                      \midrule
  \multirow{2}{*}{\textbf{SF}} & \multicolumn{2}{c}{\textbf{Baseline}} & \multicolumn{2}{c}{\textbf{Schema}} & \multicolumn{2}{c}{\textbf{Baseline}} & \multicolumn{2}{c}{\textbf{Schema}} \\
                      & \textbf{Count}             &  \textbf{Percentage}             & \textbf{Count}            &  \textbf{Percentage}            &   \textbf{Count}            &   \textbf{Percentage}             &   \textbf{Count}            &   \textbf{Percentage}           \\
  \hline                    
  \texttt{\textbf{0.1}}                 & 18           & \textbf{100}          & 18          & \textbf{100}          & 12           & \textbf{100}          & 12          & \textbf{100}          \\
  \texttt{\textbf{0.3}}                 & 16           & \textcolor{violet}{\textbf{88.9}}          & 18          & \textcolor{violet}{\textbf{100}}          & 9           & \textbf{75}         &  9          & \textbf{75}         \\
  \texttt{\textbf{1}}                   & 14           & \textcolor{violet}{\textbf{77.8}}          & 15          & \textcolor{violet}{\textbf{83.3}}         & 9           & \textbf{75}          &  9          & \textbf{75}         \\
  \texttt{\textbf{3}}                   & 11           & \textcolor{violet}{\textbf{61.1}}          & 13          & \textcolor{violet}{\textbf{72.2}}         & 7           & \textbf{58.3}          &  7          & \textbf{58.3}         \\
  \texttt{\textbf{10}}                  & 10           & \textcolor{violet}{\textbf{55.6}}          & 11          & \textcolor{violet}{\textbf{61.1}}         & 7           & \textbf{58.3}          &  7          & \textbf{58.3}         \\
  \texttt{\textbf{30}}                  & 5            & \textcolor{violet}{\textbf{27.8}}          & 7           & \textcolor{violet}{\textbf{38.9}}         & 4            & \textbf{33.3}          & 4           & \textbf{33.3}        \\ 
  \bottomrule
  \end{tabular}%
  }
  \end{table}

As shown in Tab.~\ref{tab:feasibleQueries}, the percentage of queries that did timeout increase with the scale factor for both approaches. For scale factor 30, the baseline approach could only execute 27.8\% of recursive graph queries, 
whereas the schema-based approach could execute 38.9\% recursive graph queries. Both approaches successfully executed  the same number of non-recursive graph queries across all scale factors.  
Using the schema-based approach, more recursive graph queries could be executed compared to the baseline across all scale factors, as shown in Tab.~\ref{tab:feasibleQueries}. 

We now evaluate the impact of the proposed approach on performance.

\subsection{Performance results on YAGO}\label{sec:performanceYAGO}

\begin{figure}[htbp]
  \centering
  \includegraphics[width=.75\linewidth,keepaspectratio]{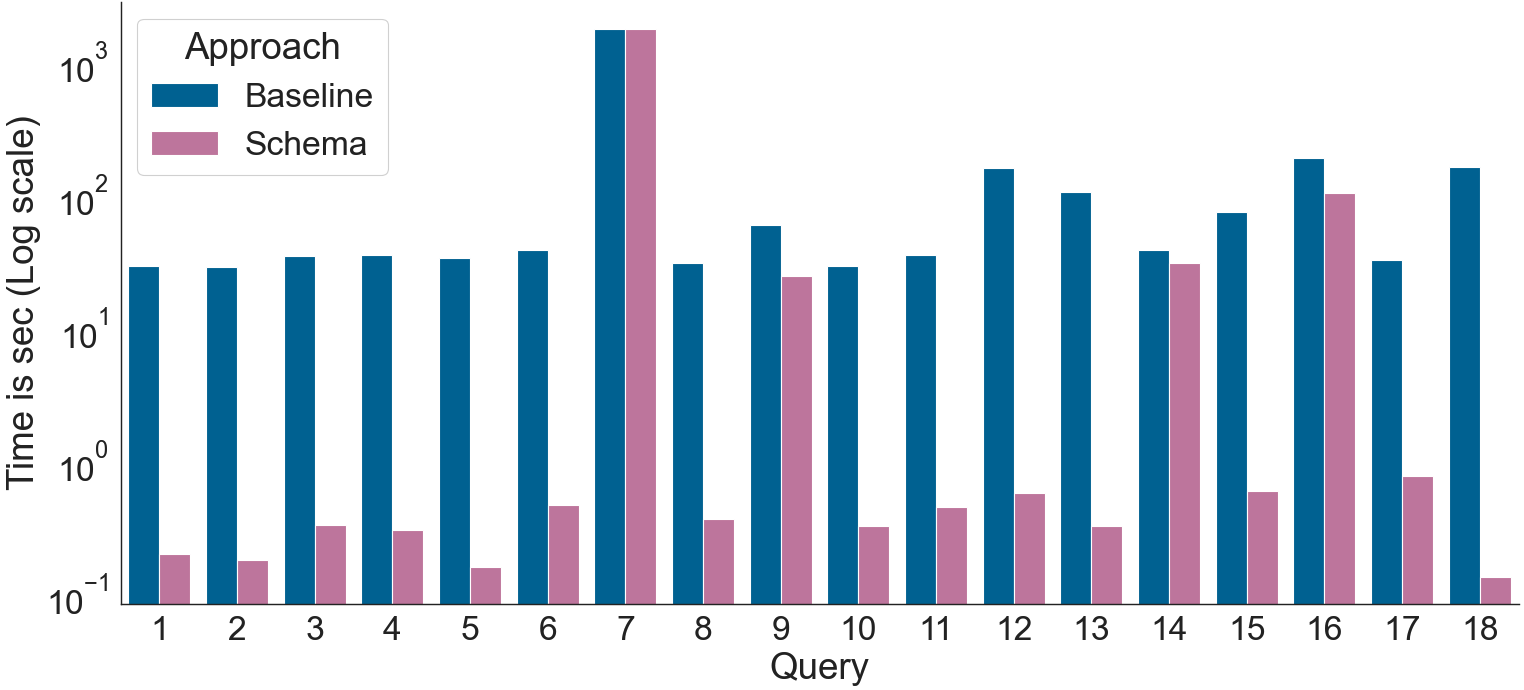}
  
  \caption{Query runtime for YAGO dataset.}
  \label{fig:yagobarplot}
\end{figure}

Fig.~\ref{fig:yagobarplot}  presents the evaluation times for the schema-based approach compared to the baseline. We observe that the schema-based approach outperforms the baseline for all YAGO queries; 
on average, YAGO queries run 6.1 times faster using the schema-based approach compared to the baseline. Notice that the time scale is logarithmic.

As mentioned in Sec.~\ref{sec:logicalQueryOpt}, the schema-enrichment approach enables the removal of 
the transitive closure operation from certain queries. Tab.~\ref{tab:fixedlengthpathstats} illustrates the replacement of the sub-path expression \texttt{isLocatedIn+} by fixed-length paths in YAGO queries.

\begin{table}[htbp]
  \centering
  \caption{Statistics on generated fixed-length paths}
  
  \resizebox{.45\linewidth}{!}{%
  \begin{tabular}{@{}ccccc@{}}
  \toprule
\textbf{YAGO Queries}                                                & \textcolor{violet}{\textbf{\#Paths}}            & \textcolor{violet}{\textbf{\texttt{Min}}}           & \textcolor{violet}{\textbf{\texttt{Avg}}}         & \textcolor{violet}{\textbf{\texttt{Max}}}          \\ 
  \midrule
\textcolor{violet}{\textbf{\texttt{1,2,3,4,5}}}                     & \textbf{1}                    & \textbf{2}             & \textbf{2}            & \textbf{2}            \\
\textcolor{violet}{\textbf{\texttt{12}}}                                & \textbf{2}                    & \textbf{1}             & \textbf{1.5}          & \textbf{2}            \\
\textcolor{violet}{\textbf{\texttt{6,8,10,11,14,15,16,17,18}}}          & \textbf{3}                    & \textbf{1}             & \textbf{2}            & \textbf{3}            \\
\textcolor{violet}{\textbf{\texttt{9}}}                                & \textbf{8}                    & \textbf{1}             & \textbf{2.5}          & \textbf{4}            \\
  \bottomrule
  \end{tabular}%
  }
  \label{tab:fixedlengthpathstats}
\end{table}

  Tab.~\ref{tab:fixedlengthpathstats} displays the total number of fixed length 
paths (\textbf{\#Paths}) along with minimum (\textbf{\texttt{Min}}), average (\textbf{\texttt{Avg}}) and maximum (\textbf{\texttt{Max}}) lengths of paths generated as replacement for transitive closure. Specifically, the transitive closure operation can be eliminated in 16 out of 18 queries for the YAGO dataset.

All (third-party) queries considered in these YAGO experiments happen to be recursive (\textbf{RQ}) graph queries. In order to further experimentally assess the impact of the approach on performance, 
we next investigate with graphs of different topology and varying sizes, in combination with recursive and non-recursive queries.

\subsection{Performance results on LDBC}\label{sec:performanceLDBC}
\nobreak 

\begin{figure}[htbp]
  \centering
  \includegraphics[width=.85\linewidth,keepaspectratio]{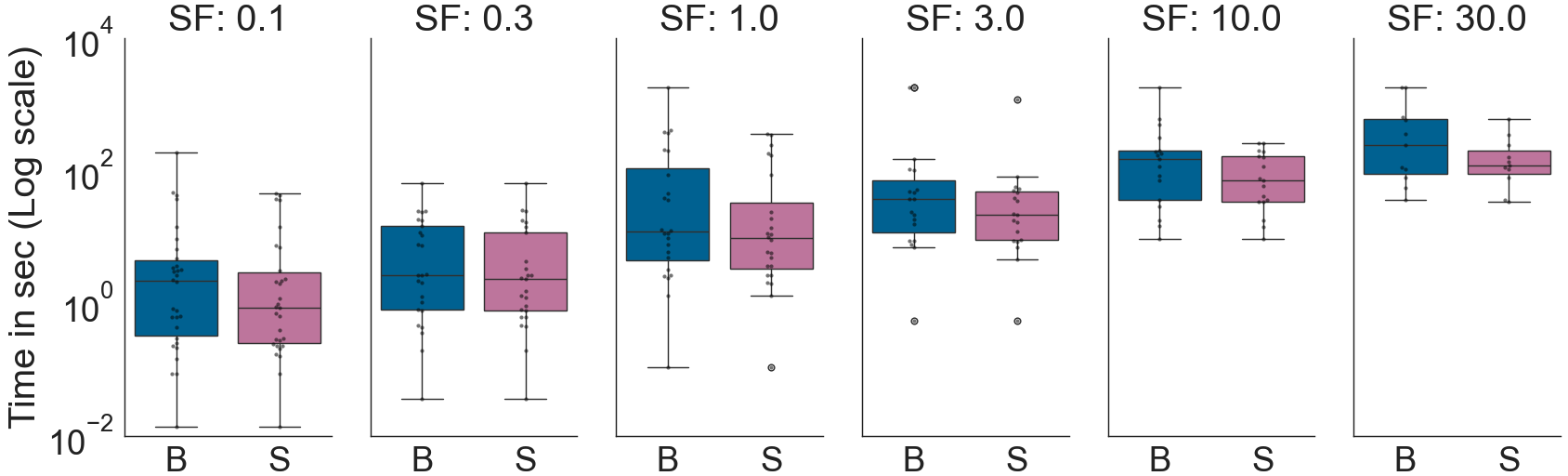}
  
  \caption{Box plot of query runtime for six scale factors (\texttt{SF=Scale Factor, B=Baseline, S=Schema})}
  \label{fig:BoxScale}
\end{figure}

We execute the 30 queries of Tab.~\ref{tab:LDBCQueries} for the six different scales of the LDBC-SNB datasets, for both the schema-based approach and the baseline. 
Therefore each query is run $6\times2=12$ times, which yields 360 query runs. Some queries timed out. 
Among the successful executions, times spent in evaluations can significantly vary from a query to another. In order to extract useful insights from  these measurements, 
we resort to statistical measures and aggregations on successful executions. 

\paragraph{\underline{Performance analysis when dataset size varies}}
We test the hypothesis that performance of query execution is improved by the schema-based approach as compared to the baseline. We conduct this test across all scale factors.
Fig.~\ref{fig:BoxScale}, presents a summary of statistics based on box plots for all the successful runs of the 30 queries of Tab.~\ref{tab:LDBCQueries} for the two approaches for graphs of increasing sizes. Running times reported on the $y$-axis use a logarithmic scale. 
This suggests that the schema-based approach is more efficient for executing queries, especially when the dataset size increases.
Notably, the performance of the schema-based approach improved after scale factor 0.3, as depicted in Fig.~\ref{fig:BoxScale}.

\paragraph{\underline{Performance analysis between recursive and non recursive queries}}
To further analyze the benefits of the schema-based approach, we categorize queries into recursive and non-recursive. Then, we collectively analyze both approaches across all six scale 
factors. Tab.~\ref{tab:fivepoint} compares schema-based and baseline (non-schema-based) approaches in query runtime. Here, \texttt{Count}, \texttt{Min}, 
\texttt{Q1}, \texttt{Q2}, \texttt{Q3}, \texttt{Max}, and \texttt{Mean} represent the total number of queries, minimum, 25$^{th}$ percentile, median, 75$^{th}$ percentile, maximum and mean query runtime in seconds, respectively.

\begin{table}[ht]
  \centering
  \caption{Query runtime summary statistics (in seconds)}
  
  \label{tab:fivepoint}
  \resizebox{.4\linewidth}{!}{%
  \begin{tabular}{@{}ccccc@{}}
  \toprule
                  & \multicolumn{2}{c}{\textbf{Recursive}} & \multicolumn{2}{c}{\textbf{Non Recursive}} \\ 
                  \midrule
                  & \textbf{Baseline}   & \textbf{Schema}  & \textbf{Baseline}     & \textbf{Schema}    \\
                  \hline 
  \texttt{\textbf{Count}}  & 78                                               & 78                                            & 48                                                 & 48                 \\
  \texttt{\textbf{Min}}    & \textbf{0.0137}                                  & \textbf{0.0137}                               & \textbf{0.087}                                     & \textbf{0.087}              \\
  \texttt{\textbf{Q1}}     & \textcolor{violet}{\textbf{2.839}}               & \textcolor{violet}{\textbf{1.968}}            & 1.211                                              & 1.271              \\
  \texttt{\textbf{Q2}} & \textcolor{violet}{\textbf{16.254}}              & \textcolor{violet}{\textbf{11.284}}           & \textcolor{violet}{\textbf{10.574}}                & \textcolor{violet}{\textbf{9.515}}              \\
  \texttt{\textbf{Q3}}     & \textcolor{violet}{\textbf{140.446}}             & \textcolor{violet}{\textbf{46.759}}           & \textcolor{violet}{\textbf{45.397}}                & \textcolor{violet}{\textbf{44.938}}             \\
  \texttt{\textbf{Max}}    & \textcolor{violet}{\textbf{1800.0}}                & \textcolor{violet}{\textbf{1171.876}}         & \textcolor{violet}{\textbf{360.887}}               & \textcolor{violet}{\textbf{353.044}}            \\ 
  \texttt{\textbf{Mean}}    & \textcolor{violet}{\textbf{213.282}}                & \textcolor{violet}{\textbf{65.244}}         & \textcolor{violet}{\textbf{46.727}}               & \textcolor{violet}{\textbf{45.398}}            \\ 
  \bottomrule
  \end{tabular}%
  }
  \end{table}

  As reported in Tab.~\ref{tab:fivepoint} comparing the schema-based approach and the baseline for recursive queries across all six scale factors, the results demonstrate that the 
  schema-based approach outperforms the baseline. Specifically, the schema-based approach is 3.26 times faster on average than the baseline for recursive queries. Both approaches exhibit comparable 
  performance when examining non-recursive queries. However, it is worth noticing that the schema-based approach generally demonstrates faster median, 75th percentile, maximum and mean query runtimes compared to the baseline.

  \begin{table}[h]
    \centering
    \caption{ Overall analysis of query runtime (in seconds)}
    
    \label{tab:fivepointOverall}
    \resizebox{.6\linewidth}{!}{%
    \begin{tabular}{cccccccc}
    \toprule
    &\textbf{\texttt{Count}} & \textbf{\texttt{Min}} & \textbf{\texttt{Q1}} & \textbf{\texttt{Median}}& \textbf{\texttt{Q3}} & \textbf{\texttt{Max}} &\textbf{\texttt{Mean}} \\
    \midrule   
    \textbf{Baseline} & 126 & \textbf{0.0137} & \textcolor{violet}{\textbf{2.6}}  & \textcolor{violet}{\textbf{14.45}} & \textcolor{violet}{\textbf{80.37}} & \textcolor{violet}{\textbf{1800.0}} &\textcolor{violet}{\textbf{149.833}}\\
    \textbf{Schema}   & 126 & \textbf{0.0137} & \textcolor{violet}{\textbf{1.38}} & \textcolor{violet}{\textbf{10.45}} & \textcolor{violet}{\textbf{46.76}} & \textcolor{violet}{\textbf{1171.87}}&\textcolor{violet}{\textbf{58.066}}\\
    \bottomrule
    \end{tabular}%
    }
    \end{table}

    Measurements in Tab.~\ref{tab:fivepointOverall} suggest that the schema-based approach consistently outperforms the baseline, with an average improvement of 2.58 times for feasible queries at all six scale factors.

  Experimental statistics show that the effect of the schema-enrichment process is even more significant for YAGO compared to LDBC.  First, YAGO queries provide more optimisation opportunities for transitive closure removal. In contrast, the transitive closure operation can only be removed in 5 out of the 30 LDBC queries. Second, YAGO queries also provide more opportunities for semi-join insertion as they use less branching and conjunction operators. Overall, 10 out of 30 LDBC queries returned to their original \ucqt{} form. In contrast, only 1 out of 18 YAGO queries, reverted to its initial form.  When the query reverts to its original form, the schema-based approach obviously has the same runtime than the baseline. This affects the overall performance impact and explains the more favorable aggregate statistics for YAGO queries.

\subsection{Evaluation on other database systems} \label{subsection:generalisability}
\nobreak

  We now conduct experiments to evaluate the schema-enrichment approach on PostgreSQL relational database system and Neo4j graph database system. We conducted experiments on the LDBC-SNB dataset, 
  focusing on~\emph{chain-shaped} queries~\cite{bonifati2020analytical}. These queries are expressed without branching and 
  conjunction operations and correspond to the formalism of the union of two-way conjunctive regular path queries (\uctworpq{}). The Cypher query language used in the Neo4j graph database only supports a restricted 
  form of \uctworpq{}~\cite{sharma2021practical}. In particular, only 15 out of 30 queries from the LDBC-SNB dataset shown in Tab.~\ref{tab:LDBCQueries} are expressible in Cypher and hence supported by Neo4j.

\begin{figure}[htbp]
  \centering
  \includegraphics[width=.85\linewidth,keepaspectratio]{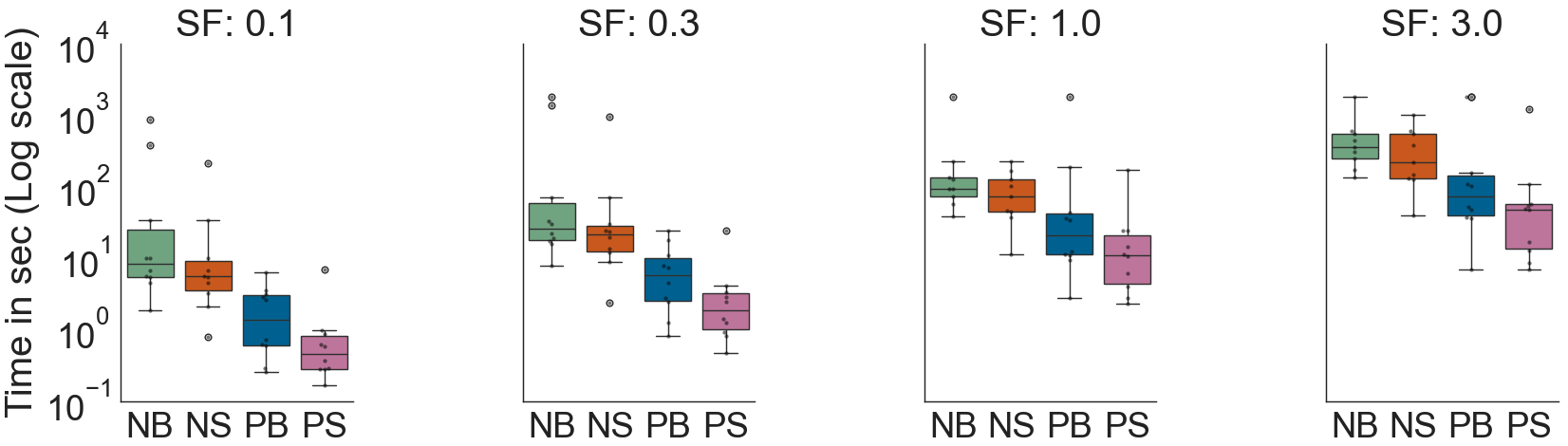}
  
  \caption{Query runtimes on Neo4j (N) and PostgreSQL (P) for LDBC-SNB. (\texttt{SF=Scale Factor, B=Baseline, S=Schema, PB=PostgreSQLBaseline})}
  \label{fig:allDBboxplotldbc}
\end{figure}

Fig.~\ref{fig:allDBboxplotldbc}, provides a comprehensive overview of the summary statistics for query runtimes using box plots specifically for the LDBC-SNB dataset at scale factors of 0.1, 0.3, 1, and 3. 
For scale factors 10, and 30, Neo4j could not complete the query evaluations within 30 minutes. Results shown in Fig.~\ref{fig:allDBboxplotldbc} suggest that the schema-based approach improves 
the performance on each individual system. Overall, the median query runtime of the schema-based approach is consistently better than the baseline across both database systems considered in this study. 
From our results we observed that PostgreSQL offers more scalability as it is capable of handling scale factors 10 and 30 (See Fig.~\ref{fig:BoxScale}), 
and yet its performance still benefits from the schema-based enrichment approach. These performance improvements are particularly important at scale factor 30. 

Measurements with Neo4j are not reported for YAGO queries as it is by far the most inefficient system for evaluating YAGO queries. This was also noticed by \cite{jachiet2020optimization}.

\paragraph{\underline{Plan-level impact of annotated path expressions}} To illustrate the plan-level impact of the schema enrichment process on annotated path expressions, 
we consider a query $\exQuery{1}$ and its schema-enriched version as query $\exQuery{2}$ in the LDBC dataset.

\begin{align*}
  \texttt{SRC, TRG} \leftarrow & \texttt{ (SRC, \textcolor{violet}{knows$\annconcat{}$workAt$\annconcat{}$isLocatedIn}, TRG)} & \exQuery{1}\\
  \texttt{SRC, TRG} \leftarrow & \texttt{ (SRC, \textcolor{violet}{knows$\annconcat{}$workAt$\annconcat{\texttt{Organisation}}$\texttt{isLocatedIn}}, TRG)} & \exQuery{2}\\
\end{align*}

The translated versions of queries $\exQuery{1}$ and $\exQuery{2}$ in SQL and Cypher are presented in Figures~\ref{listing:SQL} and~\ref{listing:Cypher}, respectively.

\lstset{language=CypherSQL, style=CypherSQL}
\begin{figure}[htbp]
 \footnotesize 
\centering
\begin{lstlisting}[mathescape, frame=single,framexleftmargin=8pt, rulecolor=\color{framegrey}]
 // SCHEMA-ENRICHED $(\exQuery{2})$                
 SELECT DISTINCT kw.Sr AS SRC, isLin.Tr AS TRG       
   FROM knows  AS kw
   JOIN workAt AS wr ON kw.Tr=wr.Sr 
   JOIN (SELECT loc.Sr AS Sr, loc.Tr AS Tr
           FROM (SELECT Sr FROM Organisation) AS org
           JOIN isLocatedIn AS loc ON loc.Sr=org.Sr         
        ) AS isLin ON wr.Tr=isLin.Sr;      
\end{lstlisting}
\vfill 
\begin{lstlisting}[mathescape, frame=single,framexleftmargin=8pt, rulecolor=\color{framegrey}]
 // BASELINE $(\exQuery{1})$                      
 SELECT DISTINCT kw.Sr AS SRC, isLin.Tr AS TRG        
   FROM knows  AS kw
   JOIN workAt AS wr ON kw.Tr=wr.Sr
   JOIN isLocatedIn AS isLin ON wr.Tr=isLin.Sr;
\end{lstlisting}

\caption{Schema-enriched and baseline queries in SQL}
\label{listing:SQL}
\end{figure}

In Fig.~\ref{listing:SQL}, the schema-enriched version of the SQL query contains an additional~\emph{semi-join} between the \texttt{Organisation} node relation and \texttt{isLocatedIn} 
edge relation (lines 5-7). The result of this~\emph{semi-join} is then combined with the \texttt{workAt} edge relation (line 8). On the other hand, in the baseline version of the query, 
additional node relation-based~\emph{semi-join} is not performed.

\lstset{language=CypherSQL, style=CypherSQL}
\begin{figure}[htbp]
 \footnotesize 
\begin{lstlisting}[mathescape, frame=single,framexleftmargin=8pt, rulecolor=\color{framegrey}] 
 // SCHEMA-ENRICHED $(\exQuery{2})$
 MATCH (SRC)-[:knows]->()-[:workAt]->(:Organisation)-[:isLocatedIn]->(TRG)
 RETURN DISTINCT SRC, TRG;
\end{lstlisting}
\vfill 
  \begin{lstlisting}[mathescape, frame=single,framexleftmargin=8pt, rulecolor=\color{framegrey}] 
 // BASELINE $(\exQuery{1})$
 MATCH (SRC)-[:knows]->()-[:workAt]->()-[:isLocatedIn]->(TRG)
 RETURN DISTINCT SRC, TRG;
\end{lstlisting}

\caption{Schema-enriched and baseline queries in Cypher}
\label{listing:Cypher}
\end{figure}

In Fig.~\ref{listing:Cypher}, the schema-enriched version of the Cypher query includes an extra node label in the graph pattern (line 2). The schema-enriched graph pattern 
indicates the selection of \texttt{isLocatedIn} labeled edges that start from nodes labeled as \texttt{Organisation} in the graph database. In contrast, the baseline version of the query in 
Cypher does not impose such constraints on the graph pattern.

\lstset{language=exeplan, style=exeplan}
\begin{figure}[htbp]
  \centering  
  \begin{minipage}[t]{\linewidth}
  \centering 
  \footnotesize 
  \begin{lstlisting}[mathescape, frame=single,framexleftmargin=8pt, rulecolor=\color{framegrey}]
 // SCHEMA-ENRICHED QUERY EXECUTION PLAN $(\exQuery{2})$  
  HashAggregate (cost = 215,265.50 rows = 2,085,899)
   Group Key: knows.Sr, islocatedin.Tr
   Planned Partitions: 32
    Hash Join (cost = 40,408.50 rows = 2,085,899)
     Hash Cond: (knows.Tr = workat.Sr)
      Seq Scan on knows (cost = 12,221.46 rows = 704,246)
      Hash (cost = 2,190.11 rows = 58,912)
       Hash Join (cost = 2,190.11 rows = 58,912)
        Hash Cond: (workat.Tr = Organisation.Sr)
         Seq Scan on workat (cost = 965.12 rows = 58,912)
         Hash (cost = 315.51 rows = 7,955)
          Merge Join (cost = 315.51 rows = 7,955)
           Merge Cond: (Organisation.Sr = islocatedin.Sr)
            Index Scan on Organisation(cost = 215.61 rows = 7,955)
            Index Scan on islocatedin(cost = 337,785.74 rows = 11,118,487)\end{lstlisting}
  \end{minipage}%
  \vfill 
  \begin{minipage}[t]{\linewidth}
  \centering 
   \footnotesize 
  \begin{lstlisting}[mathescape, frame=single,framexleftmargin=8pt, rulecolor=\color{framegrey}]
 // BASELINE QUERY EXECUTION PLAN $(\exQuery{1})$  
  HashAggregate (cost = 219,592.34 rows = 2,085,899)
   Group Key: knows.Sr, islocatedin.Tr
   Planned Partitions: 32
    Hash Join (cost = 44,735.33 rows = 2,085,899)
     Hash Cond: (knows.Tr = workat.Sr)
      Seq Scan on knows (cost = 12,221.46 rows = 704,246)
      Hash (cost = 6,516.95 rows = 58,912)
       Merge Join (cost = 6,516.95 rows = 58,912)
        Merge Cond: (islocatedin.Sr = workat.Tr)
         Index Scan on islocatedin (cost = 337,785.74 rows = 11,118,487)
         Sort (cost = 5,780.08 rows = 58,912)
          Sort Key: workat.Tr
           Seq Scan on workat (cost = 965.12 rows = 58,912)
  \end{lstlisting}
  \end{minipage}
  
  \caption{Execution plans for schema-enriched and baseline queries}
  \label{listing:exeplan}
  \end{figure}

In order to concretely illustrate the reduction of intermediate results enabled by the schema-enrichment process, Fig.~\ref{listing:exeplan} presents the schema-enriched and baseline query
execution plans for queries $\exQuery{2}$ and $\exQuery{1}$, annotated with costs and cardinalities as estimated by PostgreSQL. In Fig.~\ref{listing:exeplan}, the schema-enriched execution plan generated for query $\exQuery{2}$, indicates that the number of intermediate rows significantly decreased
after a~\emph{semi-join} is performed between the \texttt{Organisation} node relation and \texttt{isLocatedIn} edge relation (lines 12-16).
The \texttt{isLocatedIn} edge relation contains 11 million rows, however, after the~\emph{semi-join}, the number of rows reduced to approximately 8 thousand. 
In the schema-enriched execution plan, the result of the~\emph{semi-join} is combined with the \texttt{workAt} edge relation (lines 9-11), where the total number of estimated rows is 58,912 for an estimated cost of 2,190.11 (line 9). Conversely, in Fig.~\ref{listing:exeplan}, the baseline execution plan generated for query $\exQuery{1}$, the \texttt{isLocatedIn} edge relation does not undergo filtering. It is directly joined with the 
\texttt{workAt} edge relation (line 10), for an estimated higher cost of 6,516.95 (line 9). The same number of rows
is estimated for both the baseline and the schema-enriched execution plans (line 9).  

Overall, as shown in Fig.~\ref{listing:exeplan}, the schema-enriched plan (line 2) has a smaller estimated cost compared to the baseline plan (line 2) for the same number of rows in the final result set.

\section{Related Work}\label{sec:RelatedWork}
\nobreak 

Query rewriting techniques that rely on the structural information of the schema have been proposed~\cite{chakravarthy1990logic}. Authors of~\cite{abiteboul1997regular,buneman1998path,buneman2000query,vielle1989recursive} 
emphasize the schema's significance in the query rewriting. Additionally, authors in~\cite{buneman1997adding,fernandez1998optimizing} suggest that knowing the structure of the database can help reduce the query search space, 
leading to a significant improvement in overall query runtime. We now briefly discuss schema-based query rewriting techniques proposed for databases following different data models. 

\paragraph{\underline{Semi-structured databases}} Schema-based query rewriting techniques for queries over~\emph{semi-structured databases} is proposed in~\cite{abiteboul1997regular,buneman1998path,buneman2000query}. 
In semi-structured databases, schemas are expressed as~\emph{path constraints} that are regular expressions defined over edge labels~\cite{buneman2000path}. The use of path constraints to rewrite queries expressed 
in the formalism of \ctworpq{} and \uctworpq{} is proposed in~\cite{deutsch2001optimization,calvanese1998decidability,florescu1998query}. However, a significant limitation of using path constraints as a schema language 
is that the schema database consistency can only be established when path constraints are defined without using the Kleene star operator~\cite{buneman2000query}. 

\paragraph{\underline{XML databases}} Authors in~\cite{neven2006complexity,wood2003containment} propose rewriting XPath queries using the structural information stored in the schema of XML databases. For XML databases, schemas 
are expressed as~\emph{Document Type Definition} (DTDs), which are essentially regular expressions defined over the edge labels~\cite{martens2022towards}. 
Authors~\cite{miklau2004containment,benedikt2008xpath,lakshmanan2004testing,geerts2005satisfiability,hidders2004satisfiability,czerwinski2015almost} study the~\emph{satisfiability} of XPath queries 
in presence of DTDs and suggest that satisfiability is undecidable for XPath queries in presence of recursive DTDs (DTDs defined using Kleene star operator). Authors in~\cite{bressan2005accelerating,marian2003projecting} 
propose the creation of smaller XML documents by using the XPath query and structure of XML schema. Authors in~\cite{benzaken2006type} suggest a type system-based approach for XML document pruning; however, they highlight 
that creating pruned XML documents can be time-consuming and may take a similar amount of time as running the original query.

\paragraph{\underline{\SIGMODchange{Datalog}}} Schema-based query rewriting of~\emph{conjunctive queries} (\cq{}) and~\emph{union of conjunctive queries} (\ucq{}) in the presence of schema expressed as \SIGMODchange{Datalog} rules is proposed in~\cite{meier2010semantic}. 
Furthermore, authors~\cite{barcelo2020semantic} suggest that query containment of \cq{} and \ucq{} is decidable in the presence of schema expressed as non-recursive \SIGMODchange{Datalog}. Authors~\cite{bancilhon1985magic,vielle1989recursive,de2008type,schafer2010type} 
study the~\emph{containment} of \SIGMODchange{Datalog} queries in the presence of schema and suggest that the containment is decidable in the presence of non-recursive schemas. 
Regarding graph query language formalism, authors in~\cite{bonatti2004decidability,calvanese1998decidability,calvanese2005decidable} express the formalisms of \ctworpq{} and \uctworpq{} as 
\SIGMODchange{Datalog} queries and suggest that the query containment is decidable in the presence of non-recursive \SIGMODchange{Datalog} schema. 

\paragraph{\underline{Graph databases}} Graph query and schema languages have been extensively researched in the context of  knowledge graphs~\cite{wiharja2020schema,hu2022type,bellomarini2022model} with RDFS and OWL in particular~\cite{henriksson2004static,lu2007sor},  
and the standard query language SPARQL~\cite{bonifati2020sharql,haller2023query,chekol2018sparql,chekol2012sparql}. Query rewriting based on the structure of the~\emph{Resource Description Framework} (RDF) 
has been proposed in~\cite{kim2017type,abbas2017optimising} for non-recursive SPARQL queries. 

For property graphs, significant research and standardisation effort are still in progress \cite{sakr2021future,bonifati2021graph,sakr2021ensuring,bonifati2022special,bonifati2023threshold,rost2024seraph,deutsch2022graph,francis2023researcher,angles2018g}. 
So far, the lack of standard schema language hinders the application of schema-based query rewriting techniques. 
Contemporary research in property graph data models mainly focuses on property graph schema design and inference techniques~\cite{alotaibi2021property,lbath2021schema,ragab2020large,sakr2021future,bonifati2023quest,bonifati2022discopg,bonifati2019schema}. 
The design of our graph schema is motivated by existing works such as~\emph{PG-Schema} and~\emph{PG-Keys} proposed in~\cite{angles2023pg,angles2021pg}.

In~\cite{colazzo2015typing}, a type inference approach is proposed to rewrite queries expressed in the formalisms of \rpq{} and \tworpq{} using a recursive graph schema. However, the type inference system presented in~\cite{colazzo2015typing} 
is neither sound nor complete for graph query language 
formalisms containing branching and conjunction operations. Additionally, the type inference system is only explored theoretically. 

Compared to the state-of-the-art, our approach can take a $\ucqt{}$ query and a graph schema as input and generate a $\ucqt{}$ query as output, which is enhanced with structural schema information. The rewritten schema-aware query preserves the 
initial query semantics under the graph schema. 
We experimentally demonstrate that for~\emph{recursive} $\ucqt{}$s, the generated $\ucqt{}$ can be executed more efficiently using our approach.

\section{Conclusion and Perspectives}\label{sec:Conclusion}
\nobreak 

We propose a graph query rewriting method aimed at leveraging the structural information of a graph schema. The purpose is to enrich an initial query with schema information in order to improve query performance. To this end, we introduce inference rules capable of incorporating schema constraints in the path expressions contained in a query. The difficulty comes from the fact that pushing constraints through regular path expressions is complex. Our approach automates a process which would be tricky and error-prone if done manually by a developer. The soundness and completeness of the approach are proved to ensure that the initial query semantics is preserved under the schema. Furthermore, our approach is opportunistic in that it is applied only when performance gains are expected. %
We conducted extensive experiments on real and synthetic datasets, with several database systems. Experimental results show that schema-based query rewriting provides significant performance gains for recursive path queries in graphs. %
A perspective for further work is to extend the approach by considering queries with aggregrations.

\section{Acknowledgments}
This research was partially funded by the ANR \emph{GraphRec} project under grant number ``ANR-23-CE23-0010''.

\bibliographystyle{unsrtnat}
\bibliography{references}  

\begin{thebibliography}{112}
\providecommand{\natexlab}[1]{#1}
\providecommand{\url}[1]{\texttt{#1}}
\expandafter\ifx\csname urlstyle\endcsname\relax
  \providecommand{\doi}[1]{doi: #1}\else
  \providecommand{\doi}{doi: \begingroup \urlstyle{rm}\Url}\fi

\bibitem[Bell et~al.(2009)Bell, Hey, and Szalay]{bell2009beyond}
Gordon Bell, Tony Hey, and Alex Szalay.
\newblock Beyond the data deluge.
\newblock \emph{Science}, 323\penalty0 (5919):\penalty0 1297--1298, 2009.

\bibitem[Barrett et~al.(2000)Barrett, Jacob, and Marathe]{barrett2000formal}
Chris Barrett, Riko Jacob, and Madhav Marathe.
\newblock Formal-language-constrained path problems.
\newblock \emph{SIAM Journal on Computing}, 30\penalty0 (3):\penalty0 809--837,
  2000.

\bibitem[Sheth et~al.(2005)Sheth, Aleman-Meza, Arpinar, Bertram, Warke,
  Ramakrishanan, Halaschek, Anyanwu, Avant, Arpinar, et~al.]{sheth2005semantic}
Amit Sheth, Boanerges Aleman-Meza, I~Budak Arpinar, Clemens Bertram, Yashodhan
  Warke, Cartic Ramakrishanan, Chris Halaschek, Kemafar Anyanwu, David Avant,
  F~Sena Arpinar, et~al.
\newblock Semantic association identification and knowledge discovery for
  national security applications.
\newblock \emph{Journal of Database Management (JDM)}, 16\penalty0
  (1):\penalty0 33--53, 2005.

\bibitem[Yang et~al.(2020)Yang, Fan, Song, and Lin]{yang2020graph}
Fang Yang, Kunjie Fan, Dandan Song, and Huakang Lin.
\newblock Graph-based prediction of protein-protein interactions with
  attributed signed graph embedding.
\newblock \emph{BMC bioinformatics}, 21\penalty0 (1):\penalty0 1--16, 2020.

\bibitem[Sharma et~al.(2021)Sharma, Sinha, and Johnson]{sharma2021practical}
Chandan Sharma, Roopak Sinha, and Kenneth Johnson.
\newblock Practical and comprehensive formalisms for modelling contemporary
  graph query languages.
\newblock \emph{Information Systems}, 102:\penalty0 101816, 2021.

\bibitem[Alotaibi et~al.(2021)Alotaibi, Lei, Quamar, Efthymiou, and
  {\"O}zcan]{alotaibi2021property}
Rana Alotaibi, Chuan Lei, Abdul Quamar, Vasilis Efthymiou, and Fatma {\"O}zcan.
\newblock Property graph schema optimization for domain-specific knowledge
  graphs.
\newblock In \emph{2021 IEEE 37th International Conference on Data Engineering
  (ICDE)}, pages 924--935. IEEE, 2021.

\bibitem[Jachiet et~al.(2020)Jachiet, Genev{\`e}s, Gesbert, and
  Laya{\"\i}da]{jachiet2020optimization}
Louis Jachiet, Pierre Genev{\`e}s, Nils Gesbert, and Nabil Laya{\"\i}da.
\newblock On the optimization of recursive relational queries: Application to
  graph queries.
\newblock In \emph{Proceedings of the 2020 ACM SIGMOD International Conference
  on Management of Data}, pages 681--697, 2020.

\bibitem[Duschka et~al.(2000)Duschka, Genesereth, and
  Levy]{duschka2000recursive}
Oliver~M Duschka, Michael~R Genesereth, and Alon~Y Levy.
\newblock Recursive query plans for data integration.
\newblock \emph{The Journal of Logic Programming}, 43\penalty0 (1):\penalty0
  49--73, 2000.

\bibitem[Calvanese et~al.(1999)Calvanese, De~Giacomo, Lenzerini, and
  Vardi]{calvanese1999rewriting}
Diego Calvanese, Giuseppe De~Giacomo, Maurizio Lenzerini, and Moshe~Y Vardi.
\newblock Rewriting of regular expressions and regular path queries.
\newblock In \emph{Proceedings of the eighteenth ACM SIGMOD-SIGACT-SIGART
  symposium on Principles of database systems}, pages 194--204, 1999.

\bibitem[Bonifati et~al.(2018)Bonifati, Fletcher, Voigt, Yakovets, and
  Jagadish]{bonifati2018querying}
Angela Bonifati, George Fletcher, Hannes Voigt, Nikolay Yakovets, and
  HV~Jagadish.
\newblock \emph{Querying graphs}, volume~10.
\newblock Springer, 2018.

\bibitem[Angles et~al.(2017)Angles, Arenas, Barcel{\'o}, Hogan, Reutter, and
  Vrgo{\v{c}}]{angles2017foundations}
Renzo Angles, Marcelo Arenas, Pablo Barcel{\'o}, Aidan Hogan, Juan Reutter, and
  Domagoj Vrgo{\v{c}}.
\newblock Foundations of modern query languages for graph databases.
\newblock \emph{ACM Computing Surveys (CSUR)}, 50\penalty0 (5):\penalty0 1--40,
  2017.

\bibitem[Barcel{\'o} et~al.(2013)Barcel{\'o}, P{\'e}rez, and
  Reutter]{barcelo2013schema}
Pablo Barcel{\'o}, Jorge P{\'e}rez, and Juan Reutter.
\newblock Schema mappings and data exchange for graph databases.
\newblock In \emph{Proceedings of the 16th International Conference on Database
  Theory}, pages 189--200, 2013.

\bibitem[Libkin et~al.(2018)Libkin, Reutter, Soto, and
  Vrgo{\v{c}}]{libkin2018trial}
Leonid Libkin, Juan~L Reutter, Adri{\'a}n Soto, and Domagoj Vrgo{\v{c}}.
\newblock Trial: A navigational algebra for rdf triplestores.
\newblock \emph{ACM Transactions on Database Systems (TODS)}, 43\penalty0
  (1):\penalty0 1--46, 2018.

\bibitem[Reutter et~al.(2017)Reutter, Romero, and Vardi]{reutter2017regular}
Juan~L Reutter, Miguel Romero, and Moshe~Y Vardi.
\newblock Regular queries on graph databases.
\newblock \emph{Theory of computing Systems}, 61:\penalty0 31--83, 2017.

\bibitem[Vrgoc(2014)]{vrgoc2014querying}
Domagoj Vrgoc.
\newblock \emph{Querying graphs with data}.
\newblock PhD thesis, University of Edinburgh, {UK}, 2014.
\newblock URL \url{https://hdl.handle.net/1842/8953}.

\bibitem[Libkin et~al.(2013)Libkin, Martens, and
  Vrgo{\v{c}}]{libkin2013querying}
Leonid Libkin, Wim Martens, and Domagoj Vrgo{\v{c}}.
\newblock Querying graph databases with xpath.
\newblock In \emph{Proceedings of the 16th International Conference on Database
  Theory}, pages 129--140, 2013.

\bibitem[Hogan and Hogan(2020)]{hogan2020sparql}
Aidan Hogan and Aidan Hogan.
\newblock Sparql query language.
\newblock \emph{The Web of Data}, pages 323--448, 2020.

\bibitem[Francis et~al.(2018)Francis, Green, Guagliardo, Libkin, Lindaaker,
  Marsault, Plantikow, Rydberg, Selmer, and Taylor]{francis2018cypher}
Nadime Francis, Alastair Green, Paolo Guagliardo, Leonid Libkin, Tobias
  Lindaaker, Victor Marsault, Stefan Plantikow, Mats Rydberg, Petra Selmer, and
  Andr{\'e}s Taylor.
\newblock Cypher: An evolving query language for property graphs.
\newblock In \emph{Proceedings of the 2018 international conference on
  management of data}, pages 1433--1445, 2018.

\bibitem[van Rest et~al.(2016)van Rest, Hong, Kim, Meng, and
  Chafi]{van2016pgql}
Oskar van Rest, Sungpack Hong, Jinha Kim, Xuming Meng, and Hassan Chafi.
\newblock Pgql: a property graph query language.
\newblock In \emph{Proceedings of the Fourth International Workshop on Graph
  Data Management Experiences and Systems}, pages 1--6, 2016.

\bibitem[Angles et~al.(2023)Angles, Bonifati, Dumbrava, Fletcher, Green,
  Hidders, Li, Libkin, Marsault, Martens, et~al.]{angles2023pg}
Renzo Angles, Angela Bonifati, Stefania Dumbrava, George Fletcher, Alastair
  Green, Jan Hidders, Bei Li, Leonid Libkin, Victor Marsault, Wim Martens,
  et~al.
\newblock Pg-schema: Schemas for property graphs.
\newblock \emph{Proceedings of the ACM on Management of Data}, 1\penalty0
  (2):\penalty0 1--25, 2023.

\bibitem[Angles et~al.(2021)Angles, Bonifati, Dumbrava, Fletcher, Hare,
  Hidders, Lee, Li, Libkin, Martens, et~al.]{angles2021pg}
Renzo Angles, Angela Bonifati, Stefania Dumbrava, George Fletcher, Keith~W
  Hare, Jan Hidders, Victor~E Lee, Bei Li, Leonid Libkin, Wim Martens, et~al.
\newblock Pg-keys: Keys for property graphs.
\newblock In \emph{Proceedings of the 2021 International Conference on
  Management of Data}, pages 2423--2436, 2021.

\bibitem[Nguyen and Kim(2017)]{nguyen2017estimating}
Van-Quyet Nguyen and Kyungbaek Kim.
\newblock Estimating the evaluation cost of regular path queries on large
  graphs.
\newblock In \emph{Proceedings of the 8th International Symposium on
  Information and Communication Technology}, pages 92--99, 2017.

\bibitem[Yakovets et~al.(2015)Yakovets, Godfrey, and
  Gryz]{yakovets2015waveguide}
Nikolay Yakovets, Parke Godfrey, and Jarek Gryz.
\newblock Waveguide: Evaluating sparql property path queries.
\newblock In \emph{EDBT}, volume 2015, pages 525--528, 2015.

\bibitem[Agrawal(1988)]{agrawal1988alpha}
Rakesh Agrawal.
\newblock Alpha: An extension of relational algebra to express a class of
  recursive queries.
\newblock \emph{IEEE Transactions on Software Engineering}, 14\penalty0
  (7):\penalty0 879--885, 1988.

\bibitem[Gomes-Jr et~al.(2015)Gomes-Jr, Amann, and
  Santanch{\`e}]{gomes2015beta}
Luiz Gomes-Jr, Bernd Amann, and Andr{\'e} Santanch{\`e}.
\newblock Beta-algebra: Towards a relational algebra for graph analysis.
\newblock In \emph{EDBT/ICDT 2015 Joint Conference}, volume 157, 2015.

\bibitem[Leeuwen et~al.(2022)Leeuwen, Mulder, van~de Wall, Fletcher, and
  Yakovets]{leeuwen2022avantgraph}
Wilco~v Leeuwen, Thomas Mulder, Bram van~de Wall, George Fletcher, and Nikolay
  Yakovets.
\newblock Avantgraph query processing engine.
\newblock \emph{Proceedings of the VLDB Endowment}, 15\penalty0 (12):\penalty0
  3698--3701, 2022.

\bibitem[Sun et~al.(2015)Sun, Fokoue, Srinivas, Kementsietsidis, Hu, and
  Xie]{sun2015sqlgraph}
Wen Sun, Achille Fokoue, Kavitha Srinivas, Anastasios Kementsietsidis, Gang Hu,
  and Guotong Xie.
\newblock Sqlgraph: An efficient relational-based property graph store.
\newblock In \emph{Proceedings of the 2015 ACM SIGMOD International Conference
  on Management of Data}, pages 1887--1901, 2015.

\bibitem[Shanmugasundaram et~al.(1999)Shanmugasundaram, Tufte, Zhang, He,
  DeWitt, and Naughton]{shanmugasundaram1999relational}
J~Shanmugasundaram, K~Tufte, C~Zhang, G~He, DJ~DeWitt, and JF~Naughton.
\newblock Relational databases for querying xml documents: Limitations and
  opportunities, in ‘vldb’99: Proc. of the 25th int. conf. on very large
  data bases’, 1999.

\bibitem[Jindal et~al.(2014)Jindal, Rawlani, Wu, Madden, Deshpande, and
  Stonebraker]{jindal2014vertexica}
Alekh Jindal, Praynaa Rawlani, Eugene Wu, Samuel Madden, Amol Deshpande, and
  Mike Stonebraker.
\newblock Vertexica: your relational friend for graph analytics!
\newblock 2014.

\bibitem[Xirogiannopoulos et~al.(2017)Xirogiannopoulos, Srinivas, and
  Deshpande]{xirogiannopoulos2017graphgen}
Konstantinos Xirogiannopoulos, Virinchi Srinivas, and Amol Deshpande.
\newblock Graphgen: Adaptive graph processing using relational databases.
\newblock In \emph{Proceedings of the Fifth International Workshop on Graph
  Data-management Experiences \& Systems}, pages 1--7, 2017.

\bibitem[Meimaris et~al.(2017)Meimaris, Papastefanatos, Mamoulis, and
  Anagnostopoulos]{meimaris2017extended}
Marios Meimaris, George Papastefanatos, Nikos Mamoulis, and Ioannis
  Anagnostopoulos.
\newblock Extended characteristic sets: graph indexing for sparql query
  optimization.
\newblock In \emph{2017 IEEE 33rd International Conference on Data Engineering
  (ICDE)}, pages 497--508. IEEE, 2017.

\bibitem[Neumann and Moerkotte(2011)]{neumann2011characteristic}
Thomas Neumann and Guido Moerkotte.
\newblock Characteristic sets: Accurate cardinality estimation for rdf queries
  with multiple joins.
\newblock In \emph{2011 IEEE 27th International Conference on Data
  Engineering}, pages 984--994. IEEE, 2011.

\bibitem[Bornea et~al.(2013)Bornea, Dolby, Kementsietsidis, Srinivas,
  Dantressangle, Udrea, and Bhattacharjee]{bornea2013building}
Mihaela~A Bornea, Julian Dolby, Anastasios Kementsietsidis, Kavitha Srinivas,
  Patrick Dantressangle, Octavian Udrea, and Bishwaranjan Bhattacharjee.
\newblock Building an efficient rdf store over a relational database.
\newblock In \emph{Proceedings of the 2013 ACM SIGMOD International Conference
  on Management of Data}, pages 121--132, 2013.

\bibitem[Sharma and Sinha(2022)]{sharma2022flasc}
Chandan Sharma and Roopak Sinha.
\newblock Flasc: a formal algebra for labeled property graph schema.
\newblock \emph{Automated Software Engineering}, 29\penalty0 (1):\penalty0 37,
  2022.

\bibitem[Bonatti(2004)]{bonatti2004decidability}
Piero~A Bonatti.
\newblock On the decidability of containment of recursive datalog
  queries-preliminary report.
\newblock In \emph{Proceedings of the twenty-third ACM SIGMOD-SIGACT-SIGART
  symposium on Principles of database systems}, pages 297--306, 2004.

\bibitem[Calvanese et~al.(1998)Calvanese, De~Giacomo, and
  Lenzerini]{calvanese1998decidability}
Diego Calvanese, Giuseppe De~Giacomo, and Maurizio Lenzerini.
\newblock On the decidability of query containment under constraints.
\newblock In \emph{Proceedings of the seventeenth ACM SIGACT-SIGMOD-SIGART
  symposium on Principles of database systems}, pages 149--158, 1998.

\bibitem[Calvanese et~al.(2005)Calvanese, De~Giacomo, and
  Vardi]{calvanese2005decidable}
Diego Calvanese, Giuseppe De~Giacomo, and Moshe~Y Vardi.
\newblock Decidable containment of recursive queries.
\newblock \emph{Theoretical Computer Science}, 336\penalty0 (1):\penalty0
  33--56, 2005.

\bibitem[Deutsch and Tannen(2001)]{deutsch2001optimization}
Alin Deutsch and Val Tannen.
\newblock Optimization properties for classes of conjunctive regular path
  queries.
\newblock In \emph{International Workshop on Database Programming Languages},
  pages 21--39. Springer, 2001.

\bibitem[Florescu et~al.(1998)Florescu, Levy, and Suciu]{florescu1998query}
Daniela Florescu, Alon Levy, and Dan Suciu.
\newblock Query containment for conjunctive queries with regular expressions.
\newblock In \emph{Proceedings of the seventeenth ACM SIGACT-SIGMOD-SIGART
  symposium on Principles of database systems}, pages 139--148, 1998.

\bibitem[Benedikt et~al.(2008)Benedikt, Fan, and Geerts]{benedikt2008xpath}
Michael Benedikt, Wenfei Fan, and Floris Geerts.
\newblock Xpath satisfiability in the presence of dtds.
\newblock \emph{Journal of the ACM (JACM)}, 55\penalty0 (2):\penalty0 1--79,
  2008.

\bibitem[Lakshmanan et~al.(2004)Lakshmanan, Ramesh, Wang, and
  Zhao]{lakshmanan2004testing}
Laks~VS Lakshmanan, Ganesh Ramesh, Hui Wang, and Zheng Zhao.
\newblock On testing satisfiability of tree pattern queries.
\newblock In \emph{VLDB}, volume~4, pages 120--131, 2004.

\bibitem[Geerts and Fan(2005)]{geerts2005satisfiability}
Floris Geerts and Wenfei Fan.
\newblock Satisfiability of xpath queries with sibling axes.
\newblock In \emph{International Workshop on Database Programming Languages},
  pages 122--137. Springer, 2005.

\bibitem[Genev{\`{e}}s and Laya{\"{\i}}da(2006)]{geneves06}
Pierre Genev{\`{e}}s and Nabil Laya{\"{\i}}da.
\newblock A system for the static analysis of xpath.
\newblock \emph{{ACM} Trans. Inf. Syst.}, 24\penalty0 (4):\penalty0 475--502,
  2006.
\newblock \doi{10.1145/1185877.1185882}.
\newblock URL \url{https://doi.org/10.1145/1185877.1185882}.

\bibitem[Hoffart et~al.(2013)Hoffart, Suchanek, Berberich, and
  Weikum]{DBLP:journals/ai/HoffartSBW13}
Johannes Hoffart, Fabian~M. Suchanek, Klaus Berberich, and Gerhard Weikum.
\newblock {YAGO2:} {A} spatially and temporally enhanced knowledge base from
  wikipedia.
\newblock \emph{Artif. Intell.}, 194:\penalty0 28--61, 2013.
\newblock \doi{10.1016/J.ARTINT.2012.06.001}.
\newblock URL \url{https://doi.org/10.1016/j.artint.2012.06.001}.

\bibitem[Angles et~al.(2020)Angles, Thakkar, and Tomaszuk]{angles2020mapping}
Renzo Angles, Harsh Thakkar, and Dominik Tomaszuk.
\newblock Mapping rdf databases to property graph databases.
\newblock \emph{IEEE Access}, 8:\penalty0 86091--86110, 2020.

\bibitem[Francis et~al.(2023{\natexlab{a}})Francis, Gheerbrant, Guagliardo,
  Libkin, Marsault, Martens, Murlak, Peterfreund, Rogova, and
  Vrgoc]{francis2023gpc}
Nadime Francis, Am{\'e}lie Gheerbrant, Paolo Guagliardo, Leonid Libkin, Victor
  Marsault, Wim Martens, Filip Murlak, Liat Peterfreund, Alexandra Rogova, and
  Domagoj Vrgoc.
\newblock Gpc: A pattern calculus for property graphs.
\newblock In \emph{Proceedings of the 42nd ACM SIGMOD-SIGACT-SIGAI Symposium on
  Principles of Database Systems}, pages 241--250, 2023{\natexlab{a}}.

\bibitem[Hellings(2018)]{hellings2018tarski}
Jelle Hellings.
\newblock On tarski’s relation algebra: querying trees and chains and the
  semi-join algebra, 2018.

\bibitem[Hellings et~al.(2017)Hellings, Pilachowski, Van~Gucht, Gyssens, and
  Wu]{hellings2017relation}
Jelle Hellings, Catherine~L Pilachowski, Dirk Van~Gucht, Marc Gyssens, and
  Yuqing Wu.
\newblock From relation algebra to semi-join algebra: An approach for graph
  query optimization.
\newblock In \emph{Proceedings of the 16th International Symposium on Database
  Programming Languages}, pages 1--10, 2017.

\bibitem[P{\'e}rez et~al.(2006)P{\'e}rez, Arenas, and
  Gutierrez]{perez2006semantics}
Jorge P{\'e}rez, Marcelo Arenas, and Claudio Gutierrez.
\newblock Semantics and complexity of sparql.
\newblock In \emph{International semantic web conference}, pages 30--43.
  Springer, 2006.

\bibitem[Hellings et~al.(2020)Hellings, Gyssens, Wu, Van~Gucht, Van~den
  Bussche, Vansummeren, and Fletcher]{hellings2020comparing}
Jelle Hellings, Marc Gyssens, Yuqing Wu, Dirk Van~Gucht, Jan Van~den Bussche,
  Stijn Vansummeren, and George~HL Fletcher.
\newblock Comparing the expressiveness of downward fragments of the relation
  algebra with transitive closure on trees.
\newblock \emph{Information Systems}, 89:\penalty0 101467, 2020.

\bibitem[Hellings et~al.(2023)Hellings, Gyssens, Van Den~Bussche, and
  Van~Gucht]{hellings2023expressive}
Jelle Hellings, Marc Gyssens, Jan Van Den~Bussche, and Dirk Van~Gucht.
\newblock Expressive completeness of two-variable first-order logic with
  counting for first-order logic queries on rooted unranked trees.
\newblock In \emph{2023 38th Annual ACM/IEEE Symposium on Logic in Computer
  Science (LICS)}, pages 1--13. IEEE, 2023.

\bibitem[Jaakkola and Kuusisto(2023)]{jaakkola2023complexity}
Reijo Jaakkola and Antti Kuusisto.
\newblock Complexity classifications via algebraic logic.
\newblock In \emph{31st EACSL Annual Conference on Computer Science Logic, CSL
  2023}. Schloss Dagstuhl-Leibniz-Zentrum f{\"u}r Informatik, 2023.

\bibitem[Aref et~al.(2015)Aref, ten Cate, Green, Kimelfeld, Olteanu, Pasalic,
  Veldhuizen, and Washburn]{aref2015design}
Molham Aref, Balder ten Cate, Todd~J Green, Benny Kimelfeld, Dan Olteanu, Emir
  Pasalic, Todd~L Veldhuizen, and Geoffrey Washburn.
\newblock Design and implementation of the logicblox system.
\newblock In \emph{Proceedings of the 2015 ACM SIGMOD International Conference
  on Management of Data}, pages 1371--1382, 2015.

\bibitem[Leone et~al.(2006)Leone, Pfeifer, Faber, Eiter, Gottlob, Perri, and
  Scarcello]{leone2006dlv}
Nicola Leone, Gerald Pfeifer, Wolfgang Faber, Thomas Eiter, Georg Gottlob,
  Simona Perri, and Francesco Scarcello.
\newblock The dlv system for knowledge representation and reasoning.
\newblock \emph{ACM Transactions on Computational Logic (TOCL)}, 7\penalty0
  (3):\penalty0 499--562, 2006.

\bibitem[Urbani et~al.(2016)Urbani, Jacobs, and Kr{\"o}tzsch]{urbani2016vlog}
Jacopo Urbani, Ceriel~JH Jacobs, and Markus Kr{\"o}tzsch.
\newblock Vlog: A column-oriented datalog system for large knowledge graphs.
\newblock In \emph{ISWC (Posters \& Demos)}, 2016.

\bibitem[Fejza et~al.(2023)Fejza, Genev{\`e}s, Laya{\"\i}da, and
  Chlyah]{fejza2023mu}
Amela Fejza, Pierre Genev{\`e}s, Nabil Laya{\"\i}da, and Sarah Chlyah.
\newblock The $\mu$-ra system for recursive path queries over graphs.
\newblock In \emph{Proceedings of the 32nd ACM International Conference on
  Information and Knowledge Management}, pages 5041--5045, 2023.

\bibitem[Yago(2019)]{yago2019high}
YAGO Yago.
\newblock A high-quality knowledge base.
\newblock
  \url{https://www.mpi-inf.mpg.de/departments/databases-and-information-systems/research/yago-naga/yago},
  2019.

\bibitem[Erling et~al.(2015)Erling, Averbuch, Larriba-Pey, Chafi, Gubichev,
  Prat, Pham, and Boncz]{erling2015ldbc}
Orri Erling, Alex Averbuch, Josep Larriba-Pey, Hassan Chafi, Andrey Gubichev,
  Arnau Prat, Minh-Duc Pham, and Peter Boncz.
\newblock The ldbc social network benchmark: Interactive workload.
\newblock In \emph{Proceedings of the 2015 ACM SIGMOD International Conference
  on Management of Data}, pages 619--630, 2015.

\bibitem[Szárnyas(2023)]{cwi:snb}
Gábor Szárnyas.
\newblock Ldbc social network benchmark graphs.
\newblock
  \url{https://hdl.handle.net/11112/e6e00558-a2c3-9214-473e-04a16de09bf8},
  2023.

\bibitem[Suchanek et~al.(2023)Suchanek, Alam, Bonald, Paris, and
  Soria]{suchanek2023integrating}
Fabian Suchanek, Mehwish Alam, Thomas Bonald, Pierre-Henri Paris, and Jules
  Soria.
\newblock Integrating the wikidata taxonomy into yago, 2023.

\bibitem[Mhedhbi et~al.(2021)Mhedhbi, Lissandrini, Kuiper, Waudby, and
  Sz{\'a}rnyas]{mhedhbi2021lsqb}
Amine Mhedhbi, Matteo Lissandrini, Laurens Kuiper, Jack Waudby, and G{\'a}bor
  Sz{\'a}rnyas.
\newblock Lsqb: a large-scale subgraph query benchmark.
\newblock In \emph{Proceedings of the 4th ACM SIGMOD Joint International
  Workshop on Graph Data Management Experiences \& Systems (GRADES) and Network
  Data Analytics (NDA)}, pages 1--11, 2021.

\bibitem[Abul-Basher et~al.(2017)Abul-Basher, Yakovets, Godfrey,
  Ghajar-Khosravi, and Chignell]{abul2017tasweet}
Zahid Abul-Basher, Nikolay Yakovets, Parke Godfrey, Shadi Ghajar-Khosravi, and
  Mark~H Chignell.
\newblock Tasweet: optimizing disjunctive regular path queries in graph
  databases.
\newblock In \emph{EDBT/ICDT 2017 Joint Conference 20th International
  Conference on Extending Database Technology}, pages 470--473, 2017.

\bibitem[Gubichev et~al.(2013)Gubichev, Bedathur, and
  Seufert]{gubichev2013sparqling}
Andrey Gubichev, Srikanta~J Bedathur, and Stephan Seufert.
\newblock Sparqling kleene: fast property paths in rdf-3x.
\newblock In \emph{First International Workshop on Graph Data Management
  Experiences and Systems}, pages 1--7, 2013.

\bibitem[Bonifati et~al.(2020{\natexlab{a}})Bonifati, Martens, and
  Timm]{bonifati2020analytical}
Angela Bonifati, Wim Martens, and Thomas Timm.
\newblock An analytical study of large sparql query logs.
\newblock \emph{The VLDB Journal}, 29\penalty0 (2-3):\penalty0 655--679,
  2020{\natexlab{a}}.

\bibitem[Chakravarthy et~al.(1990)Chakravarthy, Grant, and
  Minker]{chakravarthy1990logic}
Upen~S Chakravarthy, John Grant, and Jack Minker.
\newblock Logic-based approach to semantic query optimization.
\newblock \emph{ACM Transactions on Database Systems (TODS)}, 15\penalty0
  (2):\penalty0 162--207, 1990.

\bibitem[Abiteboul and Vianu(1997)]{abiteboul1997regular}
Serge Abiteboul and Victor Vianu.
\newblock Regular path queries with constraints.
\newblock In \emph{Proceedings of the sixteenth ACM SIGACT-SIGMOD-SIGART
  symposium on Principles of database systems}, pages 122--133, 1997.

\bibitem[Buneman et~al.(1998)Buneman, Fan, and Weinstein]{buneman1998path}
Peter Buneman, Wenfei Fan, and Scott Weinstein.
\newblock Path constraints on semistructured and structured data.
\newblock In \emph{Proceedings of the seventeenth ACM SIGACT-SIGMOD-SIGART
  symposium on Principles of database systems}, pages 129--138, 1998.

\bibitem[Buneman et~al.(2000{\natexlab{a}})Buneman, Fan, and
  Weinstein]{buneman2000query}
Peter Buneman, Wenfei Fan, and Scott Weinstein.
\newblock Query optimization for semistructured data using path constraints in
  a deterministic data model.
\newblock In \emph{Research Issues in Structured and Semistructured Database
  Programming: 7th International Workshop on Database Programming Languages,
  DBPL’99 Kinloch Rannoch, UK, September 1--3, 1999 Revised Papers 7}, pages
  208--223. Springer, 2000{\natexlab{a}}.

\bibitem[Vielle(1989)]{vielle1989recursive}
Laurent Vielle.
\newblock Recursive query processing: The power of logic.
\newblock \emph{Theoretical computer science}, 69\penalty0 (1):\penalty0 1--53,
  1989.

\bibitem[Buneman et~al.(1997)Buneman, Davidson, Fernandez, and
  Suciu]{buneman1997adding}
Peter Buneman, Susan Davidson, Mary Fernandez, and Dan Suciu.
\newblock Adding structure to unstructured data.
\newblock In \emph{Database Theory—ICDT'97: 6th International Conference
  Delphi, Greece, January 8--10, 1997 Proceedings 6}, pages 336--350. Springer,
  1997.

\bibitem[Fernandez and Suciu(1998)]{fernandez1998optimizing}
Mary Fernandez and Dan Suciu.
\newblock Optimizing regular path expressions using graph schemas.
\newblock In \emph{Proceedings 14th International Conference on Data
  Engineering}, pages 14--23. IEEE, 1998.

\bibitem[Buneman et~al.(2000{\natexlab{b}})Buneman, Fan, and
  Weinstein]{buneman2000path}
Peter Buneman, Wenfei Fan, and Scott Weinstein.
\newblock Path constraints in semistructured databases.
\newblock \emph{Journal of Computer and System Sciences}, 61\penalty0
  (2):\penalty0 146--193, 2000{\natexlab{b}}.

\bibitem[Neven and Schwentick(2006)]{neven2006complexity}
Frank Neven and Thomas Schwentick.
\newblock On the complexity of xpath containment in the presence of
  disjunction, dtds, and variables.
\newblock \emph{Logical Methods in Computer Science}, 2, 2006.

\bibitem[Wood(2003)]{wood2003containment}
Peter~T Wood.
\newblock Containment for xpath fragments under dtd constraints.
\newblock In \emph{Database Theory—ICDT 2003: 9th International Conference
  Siena, Italy, January 8--10, 2003 Proceedings 9}, pages 300--314. Springer,
  2003.

\bibitem[Martens(2022)]{martens2022towards}
Wim Martens.
\newblock Towards theory for real-world data.
\newblock In \emph{Proceedings of the 41st ACM SIGMOD-SIGACT-SIGAI Symposium on
  Principles of Database Systems}, pages 261--276, 2022.

\bibitem[Miklau and Suciu(2004)]{miklau2004containment}
Gerome Miklau and Dan Suciu.
\newblock Containment and equivalence for a fragment of xpath.
\newblock \emph{Journal of the ACM (JACM)}, 51\penalty0 (1):\penalty0 2--45,
  2004.

\bibitem[Hidders(2004)]{hidders2004satisfiability}
Jan Hidders.
\newblock Satisfiability of xpath expressions.
\newblock In \emph{Database Programming Languages: 9th International Workshop,
  DBPL 2003, Potsdam, Germany, September 6-8, 2003. Revised Papers 9}, pages
  21--36. Springer, 2004.

\bibitem[Czerwi{\'n}ski et~al.(2015)Czerwi{\'n}ski, Martens, Parys, and
  Przybylko]{czerwinski2015almost}
Wojciech Czerwi{\'n}ski, Wim Martens, Pawel Parys, and Marcin Przybylko.
\newblock The (almost) complete guide to tree pattern containment.
\newblock In \emph{Proceedings of the 34th ACM SIGMOD-SIGACT-SIGAI Symposium on
  Principles of Database Systems}, pages 117--130, 2015.

\bibitem[Bressan et~al.(2005)Bressan, Catania, Lacroix, Li, and
  Maddalena]{bressan2005accelerating}
St{\'e}phane Bressan, Barbara Catania, Zo{\'e} Lacroix, Ying~Guang Li, and Anna
  Maddalena.
\newblock Accelerating queries by pruning xml documents.
\newblock \emph{Data \& Knowledge Engineering}, 54\penalty0 (2):\penalty0
  211--240, 2005.

\bibitem[Marian and Sim{\'e}on(2003)]{marian2003projecting}
Am{\'e}lie Marian and J{\'e}r{\^o}me Sim{\'e}on.
\newblock Projecting xml documents.
\newblock In \emph{Proceedings 2003 VLDB Conference}, pages 213--224. Elsevier,
  2003.

\bibitem[Benzaken et~al.(2006)Benzaken, Castagna, Colazzo, and
  Nguyen]{benzaken2006type}
V{\'e}ronique Benzaken, Giuseppe Castagna, Dario Colazzo, and Kim Nguyen.
\newblock Type-based xml projection.
\newblock In \emph{VLDB}, volume~6, pages 271--282, 2006.

\bibitem[Meier et~al.(2010)Meier, Schmidt, Wei, and Lausen]{meier2010semantic}
Michael Meier, Michael Schmidt, Fang Wei, and Georg Lausen.
\newblock Semantic query optimization in the presence of types.
\newblock In \emph{Proceedings of the twenty-ninth ACM SIGMOD-SIGACT-SIGART
  symposium on Principles of database systems}, pages 111--122, 2010.

\bibitem[Barcel{\'o} et~al.(2020)Barcel{\'o}, Figueira, Gottlob, and
  Pieris]{barcelo2020semantic}
Pablo Barcel{\'o}, Diego Figueira, Georg Gottlob, and Andreas Pieris.
\newblock Semantic optimization of conjunctive queries.
\newblock \emph{Journal of the ACM (JACM)}, 67\penalty0 (6):\penalty0 1--60,
  2020.

\bibitem[Bancilhon et~al.(1985)Bancilhon, Maier, Sagiv, and
  Ullman]{bancilhon1985magic}
Francois Bancilhon, David Maier, Yehoshua Sagiv, and Jeffrey~D Ullman.
\newblock Magic sets and other strange ways to implement logic programs.
\newblock In \emph{Proceedings of the fifth ACM SIGACT-SIGMOD symposium on
  Principles of database systems}, pages 1--15, 1985.

\bibitem[de~Moor et~al.(2008)de~Moor, Sereni, Avgustinov, and
  Verbaere]{de2008type}
Oege de~Moor, Damien Sereni, Pavel Avgustinov, and Mathieu Verbaere.
\newblock Type inference for datalog and its application to query optimisation.
\newblock In \emph{Proceedings of the twenty-seventh ACM SIGMOD-SIGACT-SIGART
  symposium on Principles of database systems}, pages 291--300, 2008.

\bibitem[Sch{\"a}fer and de~Moor(2010)]{schafer2010type}
Max Sch{\"a}fer and Oege de~Moor.
\newblock Type inference for datalog with complex type hierarchies.
\newblock In \emph{Proceedings of the 37th annual ACM SIGPLAN-SIGACT symposium
  on Principles of programming languages}, pages 145--156, 2010.

\bibitem[Wiharja et~al.(2020)Wiharja, Pan, Kollingbaum, and
  Deng]{wiharja2020schema}
Kemas Wiharja, Jeff~Z Pan, Martin~J Kollingbaum, and Yu~Deng.
\newblock Schema aware iterative knowledge graph completion.
\newblock \emph{Journal of Web Semantics}, 65:\penalty0 100616, 2020.

\bibitem[Hu et~al.(2022)Hu, Guti{\'e}rrez-Basulto, Xiang, Li, Li, and
  Pan]{hu2022type}
Zhiwei Hu, V{\'\i}ctor Guti{\'e}rrez-Basulto, Zhiliang Xiang, Xiaoli Li, Ru~Li,
  and Jeff~Z Pan.
\newblock Type-aware embeddings for multi-hop reasoning over knowledge graphs.
\newblock \emph{arXiv preprint arXiv:2205.00782}, 2022.

\bibitem[Bellomarini et~al.(2022)Bellomarini, Gentili, Laurenza, and
  Sallinger]{bellomarini2022model}
Luigi Bellomarini, Andrea Gentili, Eleonora Laurenza, and Emanuel Sallinger.
\newblock Model-independent design of knowledge graphs-lessons learnt from
  complex financial graphs.
\newblock In \emph{EDBT}, pages 2--524, 2022.

\bibitem[Henriksson and Ma{\l}uszy{\'n}ski(2004)]{henriksson2004static}
Jakob Henriksson and Jan Ma{\l}uszy{\'n}ski.
\newblock Static type-checking of datalog with ontologies.
\newblock In \emph{Principles and Practice of Semantic Web Reasoning: Second
  International Workshop, PPSWR 2004, St. Malo, France, September 6-10, 2004.
  Proceedings 2}, pages 76--89. Springer, 2004.

\bibitem[Lu et~al.(2007)Lu, Ma, Zhang, Brunner, Wang, Pan, and Yu]{lu2007sor}
Jing Lu, Li~Ma, Lei Zhang, Jean-S{\'e}bastien Brunner, Chen Wang, Yue Pan, and
  Yong Yu.
\newblock Sor: A practical system for ontology storage, reasoning and search.
\newblock In \emph{VLDB}, volume~7, pages 1402--1405, 2007.

\bibitem[Bonifati et~al.(2020{\natexlab{b}})Bonifati, Martens, and
  Timm]{bonifati2020sharql}
Angela Bonifati, Wim Martens, and Thomas Timm.
\newblock Sharql: Shape analysis of recursive sparql queries.
\newblock In \emph{Proceedings of the 2020 ACM SIGMOD International Conference
  on Management of Data}, pages 2701--2704, 2020{\natexlab{b}}.

\bibitem[Haller(2023)]{haller2023query}
David Haller.
\newblock A query-driven approach for shacl type inference.
\newblock In \emph{Conference on Very Large Data Bases (VLDB 2023)}, 2023.

\bibitem[Chekol et~al.(2018)Chekol, Euzenat, Genev{\`e}s, and
  Laya{\"\i}da]{chekol2018sparql}
Melisachew~Wudage Chekol, J{\'e}r{\^o}me Euzenat, Pierre Genev{\`e}s, and Nabil
  Laya{\"\i}da.
\newblock Sparql query containment under schema.
\newblock \emph{Journal on Data Semantics}, 7\penalty0 (3):\penalty0 133--154,
  2018.

\bibitem[Chekol et~al.(2012)Chekol, Euzenat, Genev{\`e}s, and
  Laya{\"\i}da]{chekol2012sparql}
Melisachew~Wudage Chekol, J{\'e}r{\^o}me Euzenat, Pierre Genev{\`e}s, and Nabil
  Laya{\"\i}da.
\newblock Sparql query containment under shi axioms.
\newblock In \emph{Proceedings of the AAAI Conference on Artificial
  Intelligence}, volume~26, pages 10--16, 2012.

\bibitem[Kim et~al.(2017)Kim, Ravindra, and Anyanwu]{kim2017type}
HyeongSik Kim, Padmashree Ravindra, and Kemafor Anyanwu.
\newblock Type-based semantic optimization for scalable rdf graph pattern
  matching.
\newblock In \emph{Proceedings of the 26th International Conference on World
  Wide Web}, pages 785--793, 2017.

\bibitem[Abbas et~al.(2017)Abbas, Genev{\`e}s, Roisin, and
  Laya{\"\i}da]{abbas2017optimising}
Abdullah Abbas, Pierre Genev{\`e}s, C{\'e}cile Roisin, and Nabil Laya{\"\i}da.
\newblock Optimising sparql query evaluation in the presence of shex
  constraints.
\newblock In \emph{BDA 2017-33{\`e}me conf{\'e}rence sur la
  {\guillemotleft}Gestion de Donn{\'e}es—Principes, Technologies et
  Applications{\guillemotright}}, pages 1--12, 2017.

\bibitem[Sakr et~al.(2021{\natexlab{a}})Sakr, Bonifati, Voigt, Iosup, Ammar,
  Angles, Aref, Arenas, Besta, Boncz, et~al.]{sakr2021future}
Sherif Sakr, Angela Bonifati, Hannes Voigt, Alexandru Iosup, Khaled Ammar,
  Renzo Angles, Walid Aref, Marcelo Arenas, Maciej Besta, Peter~A Boncz, et~al.
\newblock The future is big graphs: a community view on graph processing
  systems.
\newblock \emph{Communications of the ACM}, 64\penalty0 (9):\penalty0 62--71,
  2021{\natexlab{a}}.

\bibitem[Bonifati(2021)]{bonifati2021graph}
Angela Bonifati.
\newblock Graph processing systems back to the future.
\newblock In \emph{Proceedings of the 4th ACM SIGMOD Joint International
  Workshop on Graph Data Management Experiences \& Systems (GRADES) and Network
  Data Analytics (NDA)}, pages 1--1, 2021.

\bibitem[Sakr et~al.(2021{\natexlab{b}})Sakr, Bonifati, Voigt, and
  Iosup]{sakr2021ensuring}
Sherif Sakr, Angela Bonifati, Hannes Voigt, and Alexandru Iosup.
\newblock Ensuring the success of big graph processing for the next decade and
  beyond.
\newblock \emph{Communications of the ACM}, 64\penalty0 (9),
  2021{\natexlab{b}}.

\bibitem[Bonifati and Voigt(2022)]{bonifati2022special}
Angela Bonifati and Hannes Voigt.
\newblock Special issue on big graph data management and processing.
\newblock \emph{The VLDB Journal}, 31\penalty0 (2):\penalty0 201--202, 2022.

\bibitem[Bonifati et~al.(2023{\natexlab{a}})Bonifati, Dumbrava, Fletcher,
  Hidders, Hofer, Martens, Murlak, Shinavier, Staworko, and
  Tomaszuk]{bonifati2023threshold}
Angela Bonifati, Stefania Dumbrava, George Fletcher, Jan Hidders, Matthias
  Hofer, Wim Martens, Filip Murlak, Joshua Shinavier, Slawek Staworko, and
  Dominik Tomaszuk.
\newblock Threshold queries.
\newblock \emph{ACM SIGMOD Record}, 52\penalty0 (1):\penalty0 64--73,
  2023{\natexlab{a}}.

\bibitem[Rost et~al.(2024)Rost, Tommasini, Bonifati, Della~Valle, Rahm, Hare,
  Plantikow, Selmer, and Voigt]{rost2024seraph}
Christopher Rost, Riccardo Tommasini, Angela Bonifati, Emanuele Della~Valle,
  Erhard Rahm, Keith~W Hare, Stefan Plantikow, Petra Selmer, and Hannes Voigt.
\newblock Seraph: Continuous queries on property graph streams.
\newblock 2024.

\bibitem[Deutsch et~al.(2022)Deutsch, Francis, Green, Hare, Li, Libkin,
  Lindaaker, Marsault, Martens, Michels, et~al.]{deutsch2022graph}
Alin Deutsch, Nadime Francis, Alastair Green, Keith Hare, Bei Li, Leonid
  Libkin, Tobias Lindaaker, Victor Marsault, Wim Martens, Jan Michels, et~al.
\newblock Graph pattern matching in gql and sql/pgq.
\newblock In \emph{Proceedings of the 2022 International Conference on
  Management of Data}, pages 2246--2258, 2022.

\bibitem[Francis et~al.(2023{\natexlab{b}})Francis, Gheerbrant, Guagliardo,
  Libkin, Marsault, Martens, Murlak, Peterfreund, Rogova, and
  Vrgoc]{francis2023researcher}
Nadime Francis, Am{\'e}lie Gheerbrant, Paolo Guagliardo, Leonid Libkin, Victor
  Marsault, Wim Martens, Filip Murlak, Liat Peterfreund, Alexandra Rogova, and
  Domagoj Vrgoc.
\newblock A researchers digest of gql.
\newblock In \emph{The 26th International Conference on Database Theory, 2023},
  pages 1--1. Schloss Dagstuhl-Leibniz-Zentrum f{\"u}r Informatik,
  2023{\natexlab{b}}.

\bibitem[Angles et~al.(2018)Angles, Arenas, Barcel{\'o}, Boncz, Fletcher,
  Gutierrez, Lindaaker, Paradies, Plantikow, Sequeda, et~al.]{angles2018g}
Renzo Angles, Marcelo Arenas, Pablo Barcel{\'o}, Peter Boncz, George Fletcher,
  Claudio Gutierrez, Tobias Lindaaker, Marcus Paradies, Stefan Plantikow, Juan
  Sequeda, et~al.
\newblock G-core: A core for future graph query languages.
\newblock In \emph{Proceedings of the 2018 International Conference on
  Management of Data}, pages 1421--1432, 2018.

\bibitem[Lbath et~al.(2021)Lbath, Bonifati, and Harmer]{lbath2021schema}
Han{\^a} Lbath, Angela Bonifati, and Russ Harmer.
\newblock Schema inference for property graphs.
\newblock In \emph{EDBT 2021-24th International Conference on Extending
  Database Technology}, pages 499--504, 2021.

\bibitem[Ragab(2020)]{ragab2020large}
Mohamed Ragab.
\newblock Large scale querying and processing for property graphs phd
  symposium.
\newblock 2020.

\bibitem[Bonifati et~al.(2023{\natexlab{b}})Bonifati, Dumbrava, Martinez, and
  Mir]{bonifati2023quest}
Angela Bonifati, Stefania Dumbrava, Emile Martinez, and Nicolas Mir.
\newblock The quest for schemas in graph databases.
\newblock \emph{Looking Ahead}, 4:\penalty0 5, 2023{\natexlab{b}}.

\bibitem[Bonifati et~al.(2022)Bonifati, Dumbrava, Martinez, Ghasemi,
  Jaffr{\'e}, Luton, and Pickles]{bonifati2022discopg}
Angela Bonifati, Stefania Dumbrava, Emile Martinez, Fatemeh Ghasemi, Malo
  Jaffr{\'e}, Pac{\^o}me Luton, and Thomas Pickles.
\newblock Discopg: property graph schema discovery and exploration.
\newblock \emph{Proceedings of the VLDB Endowment}, 15\penalty0 (12):\penalty0
  3654--3657, 2022.

\bibitem[Bonifati et~al.(2019)Bonifati, Furniss, Green, Harmer, Oshurko, and
  Voigt]{bonifati2019schema}
Angela Bonifati, Peter Furniss, Alastair Green, Russ Harmer, Eugenia Oshurko,
  and Hannes Voigt.
\newblock Schema validation and evolution for graph databases.
\newblock In \emph{Conceptual Modeling: 38th International Conference, ER 2019,
  Salvador, Brazil, November 4--7, 2019, Proceedings 38}, pages 448--456.
  Springer, 2019.

\bibitem[Colazzo and Sartiani(2015)]{colazzo2015typing}
Dario Colazzo and Carlo Sartiani.
\newblock Typing regular path query languages for data graphs.
\newblock In \emph{Proceedings of the 15th Symposium on Database Programming
  Languages}, pages 69--78, 2015.

\end{thebibliography}






\end{document}